\documentclass[reqno, b5paper]{amsart}
\usepackage{amsmath}
\usepackage{amssymb}
\usepackage{amsthm}
\usepackage{enumerate}
\usepackage[mathscr]{eucal}
\usepackage{graphicx}
\usepackage{epstopdf}

\renewcommand{\part}[1]{\left(#1\right)}  
\newcommand{\parg}[1]{\left\{#1\right\}}  
\newcommand{\ket}[1]{\vert {#1} \rangle}  

\newcommand{\nl}{\par\noindent}   
\newcommand{\Tr}{{ \textrm{tr}}}       
\newcommand{\Prob}{{ \textrm{p}}}    
\newcommand{\Not}{{ \mathtt{NOT}}}

\newcommand{\SNot}{\sqrt{\mathtt{NOT}}}
\newcommand{\SId}{\sqrt{\mathtt{I}}}

\newcommand{\C}{\mathbb{C}}

\newcommand{\Hol}[1]{{\mathtt{Hol}}^\gamma({#1})}

\theoremstyle{plain}
\newtheorem{theorem}{Theorem}[section]

\newtheorem{lemma}{Lemma}[section]

\theoremstyle{definition}
\newtheorem{definition}{Definition}[section]

\theoremstyle{remark}
\newtheorem{example}{Example}

\numberwithin{equation}{section}

\begin{document}
\title[Logical arguments]{Holistic logical arguments in quantum computation}

\newcommand{\acr}{\newline\indent}

\author[M.L. Dalla Chiara \and R. Giuntini \and R. Leporini \and G. Sergioli]
{Maria Luisa Dalla Chiara* \and Roberto Giuntini** \and Roberto Leporini*** \and Giuseppe Sergioli**}

\address{\llap{*\,}Dipartimento di Lettere e Filosofia\acr
    Universit\`a di Firenze\acr
    Via Bolognese 52\acr I-50139 Firenze\acr ITALY}
\email{dallachiara@unifi.it}

\address{\llap{**\,}
Universit\`a di Cagliari\acr
    Via Is Mirrionis 1\acr I-09123 Cagliari\acr ITALY}
\email{giuntini@unica.it, giuseppe.sergioli@gmail.com}

\address{\llap{***\,}Dipartimento di Ingegneria\acr
    Universit\`a di Bergamo\acr
    viale Marconi 5\acr I-24044 Dalmine (BG)\acr ITALY}
\email{roberto.leporini@unibg.it}

\thanks{Sergioli's work has been supported by the Italian Ministry
 of Scientific Research within the FIRB project
 ``Structures and dynamics of knowledge and cognition'',
 Cagliari unit F21J12000140001; Leporini's work has been
 supported by the Italian Ministry of Scientific Research within
 the PRIN project ``Automata and Formal Languages: Mathematical
 Aspects and Applications''.}


\dedicatory{Dedicated to Anatolij Dvure\v{c}enskij}

\maketitle
\begin{abstract}
Quantum computational logics represent a logical abstraction from
the circuit-theory in quantum computation. In these logics
formulas are supposed to denote  pieces of quantum information
(qubits, quregisters or mixtures of quregisters), while logical
connectives correspond to (quantum logical) gates that transform
quantum information in a reversible way. The characteristic
holistic features of the quantum theoretic formalism (which play
an essential role in entanglement-phenomena) can be used in order
to develop a {\em holistic\/} version of the quantum computational
semantics. In contrast with the {\em compositional\/} character of
most standard semantic approaches, {\em meanings\/} of formulas
are here dealt with as {\em global\/} abstract objects that
determine the {\em contextual meanings\/} of the formulas'
components (from the {\em whole\/} to the {\em parts\/}). We
present a survey of the most significant logical arguments that
are valid or that are possibly violated in the framework of this
semantics. Some logical features that may appear {\em prima
facie\/} strange seem to reflect pretty well informal arguments
that are currently used in our rational activity.

\end{abstract}

\section{Introduction}
According to a common belief a basic aim of our use of languages
is communicating some information. There are however diverging
theories about the general concept of {\em information\/}. What
does exactly mean {\em understanding\/} or {\em interpreting\/}
the information expressed by a sentence $\alpha$ of a language
$\mathcal L$?

As is well known, the classical approaches to logic and to
information theory are based on a simple idea: the {\em
informational meaning\/} of a sentence is represented by a {\em
bit\/}, which corresponds to a classical {\em truth-value\/}
(either $1$ or $0$). At the same time, sequences of $n$ bits ({\em
registers\/}) represent possible {\em informational meanings\/} of
sequences consisting of $n$ sentences. Such a sharp and dichotomic
view of information has been put in question by a number of
non-classical theories. For instance, in the framework of the so
called ``fuzzy thinking'', {\em uncertainty\/}, {\em ambiguity\/}
and {\em vagueness\/} are currently investigated by referring to
{\em truth-degrees\/}, which have replaced the classical
truth-values {\em truth\/} and {\em falsity\/}.

{\em Quantum computational logics\/} are based on a different
idea: the information expressed by a sentence may be ambiguous
because it is stored by a quantum object, which is governed by the
indeterministic laws of quantum theory.\footnote{See, for
instance, \cite{DGL05}, \cite{Gu03}, \cite{DGG}.} Accordingly, in
these logics sentences are supposed to denote pieces of quantum
information ({\em qubits\/}, {\em quregisters\/} or, more
generally, {\em mixtures of quregisters\/}), while the logical
connectives are interpreted as {\em unitary quantum operations\/}
that transform pieces of quantum information in a reversible way.
One obtains, in this way, a great variety of {\em logical
operators\/}: some of them represent the ``quantum informational
counterparts'' of the standard connectives (like negation,
conjunction, disjunction); some others correspond to {\em
genuine\/} quantum operations that may transform classical inputs
into quantum uncertainties. In this framework, some fundamental
quantum theoretic concepts, like {\em superposition\/} and {\em
entanglement\/} (which have often been described as mysterious and
potentially paradoxical), can be used as a ``semantic resource''
for a formal analysis of theoretic situations (even far from
microphysics) where {\em ambiguity\/}, {\em holism\/} and {\em
contextuality\/} play a relevant role. In this paper we will
present a survey of the most significant logical arguments that
are valid or that are possibly violated according to a {\em
holistic\/} version of quantum computational logic. We will see
how some semantic properties of this logic, which may appear {\em
prima facie\/} somewhat strange, seem to reflect pretty well both
quantum theoretic situations and informal arguments that are
currently used in our rational activity.

\section{The mathematical environment}

It is expedient to  recall some basic   concepts of quantum
computation that play an important role in the quantum
computational semantics.\footnote{See, for instance, \cite{DGL05},
\cite{NC00}, \cite{AKN98}.} The general mathematical environment
is the $n$-fold tensor product of the Hilbert space $\C^2$:
$$\mathcal H^{(n)}:=
\underbrace{\C^2\otimes\ldots\otimes\C^2}_{n-times},$$ where all
pieces of quantum information live. The elements $\ket{1} = (0,1)$
and $\ket{0} = (1,0)$ of the canonical orthonormal basis $B^{(1)}$
of $\C^2$ represent, in this framework, the two classical bits,
which can be also regarded as the canonical truth-values
\textit{Truth} and \textit{Falsity}, respectively. The canonical basis
of $\mathcal H^{(n)}$ is the set
$$B^{(n)} =\parg{\ket{x_1}\otimes \ldots
 \otimes \ket{x_n}: \ket{x_1},\ldots, \ket{x_n} \in B^{(1)}}.$$
As usual, we will briefly write $\ket{x_1, \ldots, x_n}$ instead
of $\ket{x_1} \otimes \ldots \otimes \ket{x_n}$. By definition, a
\textit{quregister} is a unit vector of $\mathcal H^{(n)}$; while
a {\em qubit\/} is a quregister of $\mathcal H^{(1)}$. Quregisters
thus correspond to pure states, namely to maximal pieces of
information about the quantum systems that are supposed to store a
given amount of quantum information. We shall also make reference
to \textit{mixtures} of quregisters, represented by density
operators $\rho$ of $\mathcal H^{(n)}$. Of course, any quregister
$\ket{\psi}$  corresponds to a special example of density
operator: the  projection operator $P_{\ket{\psi}}$ that projects
over the closed subspace determined by $\ket{\psi}$. We will
denote by $\mathfrak D(\mathcal H^{(n)})$ the set of all density
operators of $\mathcal H^{(n)}$, while $\mathfrak D=
\bigcup_n\parg{\mathfrak D(\mathcal H^{(n)})}$ will represent the
set of all possible pieces of quantum information, briefly called
{\em qumixes}.

The choice of an orthonormal basis for the space $\C^2$ is,
obviously, a matter of convention. One can consider infinitely
many bases that are determined by the application of a unitary
operator $\mathfrak T$ to the elements of the canonical basis.
From an intuitive point of view, we can think
 that the operator
$\mathfrak T$ gives rise to a change of \textit{truth-perspective}.
While in the classical case, the truth-values \textit{Truth} and
\textit{Falsity} are identified with the two classical bits
$\ket{1}$ and $\ket{0}$, assuming a different basis corresponds to
a different idea of \textit{Truth} and
\textit{Falsity}.\footnote{Truth-perspectives play an important role
in the case of {\em epistemic quantum computational logics}. See,
for instance, \cite{DBGS} and \cite{BDCGLS}.}
 Since
any basis-change in $\C^2$ is determined by a unitary operator, we
can identify a \textit{truth-perspective} with a unitary operator
$\frak T$ of $\C^2$. We will write:
$$ \ket{1_{\frak T}}= \frak T \ket{1};\; \ket{0_{\frak T}}=\frak T\ket{0},$$ and
we will assume that $ \ket{1_{\frak T}}$ and $\ket{0_{\frak T}}$
represent, respectively, the truth-values \textit{Truth}  and
\textit{Falsity} of the truth-perspective $\frak T$. The
\textit{canonical truth-perspective} is, of course, determined by
the identity operator ${\mathtt{I}}$ of $\C^2$. We will indicate by
$B^{(1)}_\frak T$ the orthonormal basis determined by $\frak T$;
while $B^{(1)}_{\mathtt{I}}$ will represent the canonical basis. From a
physical point of view, we can suppose that each truth-perspective
is associated to an apparatus that allows one to measure a given
observable.

Any unitary operator $\frak T$ of $\mathcal H^{(1)}$ can be
naturally extended to a unitary operator $\frak T^{(n)}$ of
$\mathcal H^{(n)}$ (for any $n \geq 1$):
$$\frak T^{(n)}\ket{x_1,\ldots,x_n}=
\frak T\ket{x_1}\otimes \ldots \otimes \frak T\ket{x_n}.$$

Accordingly, any choice of a unitary operator $\frak T$ of
$\mathcal H^{(1)}$ determines an orthonormal basis $B^{(n)}_\frak
T$ for $\mathcal H^{(n)}$ such that:
$$B^{(n)}_\frak T = \parg{\frak T^{(n)}
\ket{x_1,\ldots,x_n}:\ket{x_1,\ldots,x_n} \in B_{\mathtt{I}}^{(n)}}.$$
Instead of $\frak T^{(n)}\ket{x_1,\ldots,x_n}$ we will also write
$\ket{x_{1_{\frak T}},\ldots,x_{n_{\frak T}}}$.

The elements of $B^{(1)}_\frak T$ will be called the $\frak
T$-{\textit{bits}} of $\mathcal H^{(1)}$; while the elements of
$B^{(n)}_\frak T$ will represent the $\frak T$-{\textit{registers}} of
$\mathcal H^{(n)}$. On this ground the notions of \textit{truth},
\textit{falsity} and \textit{probability} with respect to any
truth-perspective $\mathfrak T $ can be defined in a natural way.

\begin{definition} \textit{($\frak T$-true and $\frak T$-false registers)}
\label{def:true}\nl
\begin{itemize}
\item $\ket{x_{1_{\frak T}},\ldots,x_{n_{\frak T}}}$ is a
    \textit{$\frak T$-true register} iff $\ket{x_{n_{\frak T}}}
    =\ket{1_{\frak T}};$
\item $\ket{x_{1_{\frak T}},\ldots,x_{n_{\frak T}}}$ is a
    \textit{$\frak T$-false register} iff $\ket{x_{n_{\frak
    T}}} =\ket{0_{\frak T}}.$
\end{itemize}
\end{definition}

In other words, the \textit{$\frak T$-truth-value} of a $\frak
T$-register (which corresponds to a sequence of $\frak T$-bits) is
determined by its last element.
\begin{definition} \textit{($\mathfrak T$-truth and $\mathfrak T$-falsity)}\nl
\begin{itemize}
\item The $\mathfrak T$-\textit{truth} of $\mathcal H^{(n)}$ is
    the projection operator $^\mathfrak TP_1^{(n)}$ that
    projects over the closed subspace spanned by the set of
    all $\mathfrak T$- true registers;
\item the $\mathfrak T$-\textit{falsity} of $\mathcal H^{(n)}$
    is the projection operator $^\mathfrak TP_0^{(n)}$ that
    projects over the closed subspace spanned by the set of
    all $\mathfrak T$- false registers.
\end{itemize}
\end{definition}

In this way, truth and falsity are dealt with as mathematical
representatives of possible physical properties. Accordingly, by
applying the Born-rule, one can naturally define the
probability-value of any qumix with respect to the
truth-perspective $\mathfrak T$.

\begin{definition} \textit{(}$\mathfrak T$-\textit{Probability)}\nl
For any $\rho \in \mathfrak D(\mathcal H^{(n)})$,
$${\Prob}_\mathfrak T(\rho) := {\Tr}(^\mathfrak TP_1^{(n)} \rho),$$
where ${\Tr}$ is the trace-functional.
\end{definition}
We interpret ${\Prob}_\mathfrak T(\rho)$ as the probability that
the information $\rho$ satisfies the $\mathfrak T$-Truth. In the
particular case of qubits, we will obviously obtain:
$${\Prob}_\mathfrak T(a_0\ket{0_\mathfrak T} + a_1\ket{1_\mathfrak T}) =|a_1|^2.$$

For any choice of a truth-perspective  $\mathfrak T$, the set
$\mathfrak D$ of all density operators can be pre-ordered by a
relation that is defined in  terms of the probability-function
${\Prob}_\mathfrak T$. In Section 4 we will see how this relation
will play an important semantic role.

\begin{definition}\textit{(Preorder)} \label{de:preordine}\nl
$\rho \preceq_\mathfrak T \sigma$ iff ${\Prob}_\mathfrak T(\rho)
\le {\Prob}_\mathfrak T(\sigma)$.
\end{definition}

When $\mathfrak T$ is the canonical truth-perspective $\mathtt{I}$,
we will also write: $P_1^{(n)}$,  $P_0^{(n)}$,  $\mathtt{p}$,
$\preceq$ (instead of $^\mathtt{I}P_1^{(n)}$,  $^\mathtt{I}P_0^{(n)}$,
$\mathtt{p}_\mathtt{I}$, $\preceq_\mathtt{I}$).

As is well known,  {\em entanglement\/} represents one of the most
crucial (and to a certain extent ``mysterious'')  feature of
quantum theory.  Consider a composite system $S= S_1 + \ldots +
S_t$ and its Hilbert space $\mathcal H^{(n)} = \mathcal
H^{(n_1)}\otimes \ldots \otimes \mathcal H^{(n_t)}$. Let $\rho \in
\mathfrak D(\mathcal H^{(n)}$) be a state of $S$ and let $i_1,
\ldots, i_r \in
\parg{1,\ldots,t}$. The quantum theoretic formalism determines the
{\em reduced state\/} of $\rho$ with respect to the subsystem
$S_{i_1} + \ldots + S_{i_r}$. We will indicate this state by
$Red^{(i_1, \ldots, i_r)}_{[n_1,\ldots,n_t]}(\rho)$.

 It is expedient to recall a characteristic property of reduced states
(described by the following Lemma).
\begin{lemma} \label{le:reduced}\nl\nl
 Let $\rho \in \mathfrak
D(\mathcal H^{(m)} \otimes
    \mathcal H^{(p)})$. The reduced state $Red^{(1)}_{[m,p]}$ is the
    unique operation of $\mathfrak D(\mathcal H^{(m)} \otimes
    \mathcal H^{(p)})$ into $\mathfrak D(\mathcal H^{(m)})$
    such that for any self-adjoint operator $A^{(m)}$ of $\mathcal
    H^{(m)}$ and for any $\rho \in \mathfrak D(\mathcal
    H^{(m)} \otimes  \mathcal H^{(p)})$:
$$
\mathtt{tr}((A^{(m)}\otimes \mathtt{I}^{(p)})\rho)=
\mathtt{tr}(A^{(m)} Red^{(1)}_{[m,p]}(\rho)).
$$

A similar relation holds  for the reduced state
$Red^{(2)}_{[m,p]}.$
\end{lemma}

 A characteristic situation  that arises in
entanglement-phenomena is the following: while the state of the
global system is pure (a maximal information), the reduced states
of some  subsystems   are  mixtures (non-maximal pieces of
information). Hence our information about the {\em whole\/} cannot
be reconstructed as a function of our pieces of information about
the {\em parts\/}. Although entanglement can be defined both for
pure and for mixed states, in  this article we will be only
concerned with entangled quregisters.

\begin{definition} {\em ($t$-partite entangled quregister)}
\label{de:npartito}\nl A quregister $\ket{\psi}$ of $\mathcal
H^{(n)}=\mathcal H^{(n_1)}\otimes \ldots \otimes \mathcal
H^{(n_t)}$ is called a {\em $t$-partite entangled state\/}
 iff all reduced states
  $Red^{(1)}_{[n_1,\ldots,n_t]}(P_{\ket{\psi}}),\ldots,
  Red^{(t)}_{[n_1,\ldots,n_t]}(P_{\ket{\psi}})$ are proper mixtures.
\end{definition}

As a consequence a  $t$-partite entangled quregister cannot be
represented as
 a tensor product of the reduced states of its parts.
When all reduced states
$Red^{(i)}_{[n_1,\ldots,n_t]}(P_{\ket{\psi}})$ are the qumix
$\frac{1}{2^{n_i}}{\mathtt{I}}^{(n_i)}$ (which represents a perfect
 ambiguous information) one says that $\ket{\psi}$ is a {\em
 $t$-partite
 maximally entangled state\/}.

\begin{definition} {\em (Entangled quregister with respect to some parts)}
\label{de:partentangl}\nl A quregister $\ket{\psi}$ of $\mathcal
H^{(n)} =  \mathcal H^{(n_1)}\otimes \ldots \otimes \mathcal
H^{(n_t)}$ is called {\em entangled\/}
 with respect to its parts labelled by the
indices $i_1, \ldots,i_r$ (with $1 \le i_1,\ldots,i_r \le t$) iff
 the  reduced states $Red^{(i_1)}_{[n_1,\ldots,n_t]}(P_{\ket{\psi}}),\ldots,
    Red^{(i_r)}_{[n_1,\ldots,n_t]}(P_{\ket{\psi}})$ are proper mixtures.

\end{definition}

Since the notion of reduced state is independent of the choice of
a particular basis, it turns out that the status  of {\em
$t$-partite entangled quregister\/},
 {\em maximally entangled quregister\/}  and {\em entangled quregister with
 respect to some parts\/} is invariant under changes of
 truth-perspective.

\begin{example}\nl

\begin{itemize}
\item The quregister
$$ \ket{\psi} = \frac{1}{\sqrt{2}}(\ket{0,0,0} + \ket{1,1,1}) $$
is a 3-partite maximally entangled quregister of $\mathcal
H^{(3)}$;
\item the quregister
$$ \ket{\psi} = \frac{1}{\sqrt{2}}(\ket{0,0,0} + \ket{1,1,0}) $$
is an entangled quregister  of $\mathcal H^{(3)}$  with
respect to its first and second part.
\end{itemize}
\end{example}

\section{Quantum logical gates and the holistic conjunction}
As is well known, quantum information is processed by
\textit{quantum logical gates} (briefly, \textit{gates}): unitary
operators that transform quregisters into quregisters in a
reversible way. Let us recall the definition of  some gates that
play a special role both from the computational and from the
logical point of view.

\begin{definition}\label{de:not}{\em (The negation)}
\nl For any $n\geq 1$, the {\em negation} on $\mathcal H^{(n)}$ is
the linear operator ${\Not}^{(n)}$ such that, for every element
$\ket{x_1,\ldots,x_n}$ of the  canonical basis,
$$
{\Not}^{(n)} \ket{x_1,\ldots,x_n} = \ket{x_1,\ldots, x_{n-1}}\otimes
\ket{1-x_{n}}.
$$
\end{definition}

In particular, we obtain:
$${\Not}^{(1)}\ket{0}= \ket{1}; \quad
{\Not}^{(1)}\ket{1}= \ket{0}, $$ according to the classical
truth-table of negation.

\begin{definition} {\em (The Toffoli-gate)}\label{de:toffoli}\nl
For any $m,n,p\geq 1$, the {\em Toffoli-gate\/} is the linear
operator ${\mathtt{T}}^{(m,n,p)}$ defined on $\mathcal H^{(m+n+p)}$
such that, for every element $\ket{x_1,\ldots, x_m}\otimes
\ket{y_1,\ldots, y_n}\otimes
       \ket{z_1,\ldots, z_p} $
of the  canonical basis,
\begin{multline*}
{\mathtt{T}}^{(m,n,p)} \ket{x_1,\ldots, x_m, y_1,\ldots,
y_n, z_1,\ldots, z_p}\\
=\ket{x_1,\ldots, x_m, y_1,\ldots, y_n,z_1,\ldots,
z_{p-1}}\otimes\ket{x_m y_n \widehat{+} z_p},
\end{multline*}
where $\widehat{+}$ represents the addition modulo $2$.
\end{definition}

The following Lemma  asserts a characteristic property of the
Toffoli-gate (which turns out to be useful from the computational
point of view).

\begin{lemma}\cite{DCGFLS}\label{le:tof}\nl\nl

$\mathtt{T}^{(m,n,1)}= [(\mathtt{I}^{(m+n)}-P_1^{(m)}\otimes
P_1^{(n)})\otimes \mathtt{I}^{(1)}]\, + \,[P_1^{(m)}\otimes
P_1^{(n)}\otimes \mathtt{NOT}^{(1)}]$.
\end{lemma}


 \begin{definition} {\em (The ${\mathtt{XOR}}$-gate)} \label{de:xor}\nl
For any $m,n\geq 1$, the {\em ${\mathtt{XOR}}$-gate\/} is the linear
operator ${\mathtt{XOR}}^{(m,n)}$ defined on $\mathcal H^{(m+n)}$ such
that, for every element $\ket{x_1,\ldots, x_m}\otimes
\ket{y_1,\ldots, y_n}$ of the  canonical basis,
$${\mathtt{XOR}}^{(m,n)} \ket{x_1,\ldots, x_m, y_1,\ldots,y_n}
=\ket{x_1,\ldots, x_m, y_1,\ldots, y_{n-1}}\otimes\ket{x_m \widehat{+} y_n}.$$

 \end{definition}

\begin{definition}{\em (The Hadamard-gate)}
\nl For any $n\geq 1$, the {\em Hadamard-gate} on $\mathcal
H^{(n)}$ is the linear operator $\SId^{(n)}$ such that for every
element $\ket{x_1,\ldots,x_n}$ of the  canonical basis:
\begin{equation*}
\SId^{(n)}\ket{x_1,\ldots,x_n}=
 \ket{x_1,\ldots,x_{n-1}}\otimes
  \frac{1}{\sqrt{2}}\left((-1)^{x_{n}}\ket{x_n} +\ket{1-x_{n}}\right).
\end{equation*}

  \end{definition}

In particular we obtain:
$$\SId^{(1)} \ket{0}= \frac{1}{\sqrt{2}}(\ket{0} + \ket{1});
\SId^{(1)} \ket{1}= \frac{1}{\sqrt{2}}(\ket{0} - \ket{1}). $$
Hence, $\SId^{(1)}$ transforms bits into genuine qubits.

\begin{definition}{\em (The square root of $\Not$)}\nl
For any $n\geq 1$, the {\em square root of $\Not$} on $\mathcal
H^{(n)}$ is the linear operator $\sqrt{\Not}^{(n)}$ such that
for every element $\ket{x_1,\ldots,x_n}$ of the canonical basis:
\begin{equation*}
\sqrt{\Not}^{(n)}\ket{x_1,\ldots,x_n}=
\ket{x_1,\ldots,x_{n-1}}\otimes \left(\frac{1-i}{2}\ket{x_n} +\frac{1+i}{2}\ket{1-x_{n}}\right),
\end{equation*}
where $i=\sqrt{-1}$.
\end{definition}

All gates can be naturally transposed from the canonical
truth-perspective to any truth-perspective  $\mathfrak T$. Let
$G^{(n)}$ be any gate defined with respect to the canonical
truth-perspective. The {\em twin-gate\/} $G^{(n)}_{\mathfrak T}$,
defined with respect to the truth-perspective $\mathfrak T$, is
determined as follows:
$$
G^{(n)}_{\mathfrak T}:=
 \mathfrak T^{(n)} G^{(n)}\mathfrak T^{{(n)}^\dagger},   $$
 where
$\mathfrak T^{{(n)}^\dagger}$ is the adjoint of $\mathfrak T$.

All $\mathfrak T$-gates can be canonically extended to the set
$\mathfrak D$ of all qumixes.  Let $G_\mathfrak T$ be any gate
defined on $\mathcal H^{(n)}$. The corresponding {\em qumix
gate\/} (also called {\em unitary quantum operation\/})
$^\mathfrak DG_\mathfrak T$ is defined as follows for any $\rho
\in \mathfrak D(\mathcal H^{(n)})$:
$$^{\mathfrak D}G_\mathfrak T\rho=
 G_\mathfrak T\,\rho\, G_\mathfrak T^\dagger.  $$
 For the sake simplicity, also the qumix gates $^\mathfrak DG_\mathfrak
 T$ will be briefly called  {\em gates}.

 The Toffoli-gate
 $^{\mathfrak D}{\mathtt{T}}^{(m,n,p)}_\mathfrak T$ allows us to define a
 reversible operation ${\mathtt{AND}}_\mathfrak T^{(m,n)}$ that represents
  a {\em
 holistic conjunction\/}.

 \begin{definition} {\em (The holistic
 conjunction)}\label{de:and}\nl\nl
For any $m,n\geq 1$  the {\em holistic conjunction \/} ${\mathtt{AND}}_\mathfrak T^{(m,n)}$  with respect to the truth-perspective
$\mathfrak T$ is defined as follows for any qumix $\rho \in
\mathfrak D(\mathcal H^{(m+n)})$:
$${\mathtt{AND}}_\mathfrak T^{(m,n)}(\rho):=\,^\mathfrak D {\mathtt{T}}^{(m,n,1)}_\mathfrak T
(\rho \otimes \,
^\mathfrak T P_0^{(1)}),  $$

where the $\frak T$-falsity  $^\mathfrak T P_0^{(1)}$  plays the
role of an {\em ancilla}.
\end{definition}

When $\mathfrak T = {\mathtt{I}}$, we will write ${\mathtt{AND}}^{(m,n)}$
(instead of ${\mathtt{AND}}_{{\mathtt{I}}}^{(m,n)}$).

It is worth-while noticing that generally
$$ {\mathtt{AND}}_\mathfrak T^{(m,n)}(\rho)\,\neq\,
{\mathtt{AND}}_\mathfrak T^{(m,n)}(Red^{(1)}_{[m,n]}(\rho)\otimes
Red^{(2)}_{[m,n]}(\rho)).  $$ Roughly, we might say that the holistic
conjunction defined on a global information consisting of two
parts  does not generally  coincide with the conjunction of the
two separate parts. As an example, we can consider the following
qumix (which corresponds to a maximally entangled pure state):

$$\rho = P_{\frac {1}{\sqrt{2}}(\ket{0,0} + \ket{1,1})}.$$
We have:
$$ {\mathtt{AND}}^{(1,1)}(\rho) = \, ^\mathfrak D{\mathtt{T}}^{(1,1,1)}
 (P_{\frac {1}{\sqrt{2}}(\ket{0,0} + \ket{1,1})} \otimes P^{(1)}_0) \,=\,
  P_{\frac {1}{\sqrt{2}}(\ket{0,0,0} + \ket{1,1,1})}\quad, $$ which
also represents a maximally entangled quregister.

At the same time we have:
$${\mathtt{AND}}^{(1,1)}(Red^{(1)}_{[1,1]}(\rho)\otimes  Red^{(2)}_{[1,1]}(\rho)) =
{\mathtt{AND}}^{(1,1)}(\frac {1}{2}{\mathtt{I}}^{(1)} \otimes
 \frac {1}{2}{\mathtt{I}}^{(1)}),  $$
 which is a proper mixture.

 Furthermore, we have:
 $${\mathtt{p}}({\mathtt{AND}}^{(1,1)}(\rho))= \frac {1}{2};\quad
 {\mathtt{p}}({\mathtt{AND}}^{(1,1)}
 (Red^{(1)}_{[1,1]}(\rho)\otimes  Red^{(2)}_{[1,1]}(\rho)))= \frac {1}{4}.  $$


We will now investigate some interesting probabilistic properties
of the holistic conjunction (illustrated by the following Theor.
\ref{th:prand} and Theor. \ref{th:and}).

Let us first recall that the set of all projection operators of a
Hilbert space $\mathcal H^{(n)}$ is partially ordered by the
following relation:
$$
P\le Q\quad\text{iff} \quad PQ=P.$$ We have: $P\le
Q\quad\text{iff} \quad\Tr({P\rho})\le\Tr(Q\rho) \quad \text{for
any} \quad\rho\in\mathfrak D(\mathcal H^{(n)}).
$

\begin{lemma}\nl \label{le:project}
\begin{enumerate}
\item  [(1)]$P_1^{(m)}\otimes P_1^{(n)} \le P_1^{(m)}\otimes
    \mathtt{I}^{(n)}$.
\item  [(2)]$P_1^{(m)}\otimes P_1^{(n)} \le \mathtt{    I}^{(m)}\otimes P_1^{(n)}$.
\item [(3)] For any $\rho\in\mathfrak D( \mathcal
    H^{(m)}\otimes \mathcal H^{(n)})$:
\begin{enumerate}
\item [(a)]$\mathtt{p}(\rho)=\mathtt{    p}(Red^{(2)}_{[m+n-1,1]}(\rho))$;
\item [(b)]$\mathtt{tr}((P_1^{(m)}\otimes \mathtt{    I}^{(n)})\rho)= \mathtt{    tr}(P_1^{(1)}Red^{(2)}_{[m-1,1,n]}(\rho))$.

\end{enumerate}
\end{enumerate}
\end{lemma}

\begin{proof}\
\begin{enumerate}
\item [(1)] $(P_1^{(m)}\otimes P_1^{(n)})(P_1^{(m)}\otimes
    {\mathtt{I}}^{(n)})=(P_1^{(m)}P_1^{(m)})\otimes
    (P_1^{(n)}{\mathtt{I}}^{(n)})=P_1^{(m)}\otimes P_1^{(n)} $.
\item  [(2)]Similar to $(1)$.
\item [(3)]
\begin{enumerate}
\item [(a)] ${\mathtt{p}}(\rho)=\Tr(P_1^{(m+n)}\rho) =\Tr(({\mathtt{I}}
    ^{(m+n-1)}\otimes P_1^{(1)})\rho)$ \nl
    $=\Tr(P_1^{(1)}Red^{(2)}_{[m+n-1,1]}(\rho))={\mathtt{p}}(Red^{(2)}_{[m+n-1,1]}(\rho))$.

\item [(b)] $\Tr((P_1^{(m)}\otimes {\mathtt{I}}^{(n)})\rho)=\Tr(({\mathtt{I}}^{(m-1)}\otimes
    P_1^{(1)}\otimes {\mathtt{I}}^{(n)})\rho)$\nl
    $=\Tr((P_1^{(1)}\otimes {\mathtt{I}}^{(n)})Red^{(2,3)}_{[m-1,1,n]}(\rho))=\Tr(P_1^{(1)}
    Red^{(2)}_{[m-1,1,n]}(\rho))$.
\end{enumerate}
\end{enumerate}
\end{proof}

\begin{theorem}\label{th:prand}\nl \nl
For any $\rho   \in \mathfrak D(\mathcal H^{(m+n)}),
 \mathtt{p}(\mathtt{AND}^{(m,n)}(\rho))= \mathtt{tr}((P_1^{(m)}\otimes
P_1^{(n)})\rho).$
\begin{proof}\nl \nl
 By definition of ${\mathtt{AND}}^{(m,n)}$ and by Lemma \ref{le:tof}:
\begin{align*}
  &{\mathtt{AND}}^{(m,n)}(\rho) =
{\mathtt{T}}^{(m,n,1)}(\rho\otimes P_0^{(1)}){\mathtt{T}}^{(m,n,1)} \\
  &= [({\mathtt{I}}^{(m+n)}-P_1^{(m)}\otimes P_1^{(n)})\otimes {\mathtt{I}}^{(1)}]
(\rho\otimes P_0^{(1)})[({\mathtt{I}}^{(m+n)}-P_1^{(m)}\otimes P_1^{(n)})
\otimes {\mathtt{I}}^{(1)}]+\\
  &+[P_1^{(m)}\otimes P_1^{(n)}\otimes {\Not}^{(1)}]
  (\rho\otimes P_0^{(1)})[P_1^{(m)}\otimes P_1^{(n)}\otimes {\Not}^{(1)}].\\
 & \text{One can easily see that}\\
 & P_1^{(m+n+1)}({\mathtt{I}}^{(m+n)}-P_1^{(m)}\otimes P_1^{(n)})\otimes
  {\mathtt{I}}^{(1)}
(\rho\otimes P_0^{(1)})({\mathtt{I}}^{(m+n)}-P_1^{(m)}\otimes
P_1^{(n)})\otimes {\mathtt{I}}^{(1)}\\
& \text{is the null projection operator. Consequently:}\\
&{\mathtt{p}}({\mathtt{AND}}^{(m,n)}(\rho))=\\
&=\Tr(P_1^{(m+n+1)}(P_1^{(m)}
\otimes P_1^{(n)}\otimes {\Not}^{(1)})(\rho
\otimes P_0^{(1)})(P_1^{(m)}\otimes P_1^{(n)}\otimes {\Not}^{(1)}))=\\
&=\Tr(P_1^{(m+n+1)}((P_1^{(m)}\otimes P_1^{(n)})\rho (P_1^{(m)}
\otimes P_1^{(n)}))\otimes {\Not}^{(1)}P_0^{(1)}{\Not}^{(1)}))=\\
& = \Tr ((P_1^{(m)}\otimes P_1^{(n)}) \rho (P_1^{(m)}\otimes P_1^{(n)})) \Tr
(P_1^{(1)}P_1^{(1)})\\
&=\Tr((P_1^{(m)}\otimes P_1^{(n)})\rho).
\end{align*}
\end{proof}
\end{theorem}

\begin{theorem}\label{th:and}\nl
For any $\rho \in \mathfrak D(\mathcal H^{(m)} \otimes  \mathcal
H^{(n)})$:
$$\mathtt{p}(\mathtt{AND}^{(m,n)}(\rho))
\le \mathtt{p}(Red^{(1)}_{[m,n]}(\rho))\quad\text{and}\quad
\mathtt{p}({\mathtt{AND}}^{(m,n)}(\rho))
\le \mathtt{p}(Red^{(2)}_{[m,n]}(\rho)).$$

\begin{proof} \nl \nl
 Let $\rho   \in \mathfrak D(\mathcal
H^{(m+n)})$. By Lemma \ref{le:project} (1) we have:
$$P_1^{(m)}\otimes P_1^{(n)}\le P_1^{(m)}\otimes {\mathtt{I}}^{(n)}.$$
 Hence, by Lemma \ref{le:project} (3):
$$\Tr((P_1^{(m)}\otimes P_1^{(n)})\rho)\le \Tr((P_1^{(m)}\otimes {\mathtt{I}}^{(n)})\rho)=
\Tr(P_1^{(1)}Red_{[m-1,1,n]}^{(2)}(\rho))= {\mathtt{p}}(Red^{(1)}_{[m,n]}(\rho)).$$  Since ${\mathtt{p}}({\mathtt{AND}}^{(m,n)}(\rho))=\Tr((P_1^{(m)}\otimes P_1^{(n)})\rho)$ (by
Theorem \ref{th:prand}), we obtain:
$${\mathtt{p}}({\mathtt{AND}}^{(m,n)}(\rho))\le {\mathtt{p}}(Red^{(1)}_{[m,n]}(\rho)).$$

 In a similar way one can prove
that:
$${\mathtt{p}}({\mathtt{AND}}^{(m,n)}(\rho))\le {\mathtt{p}}(Red^{(2)}_{[m,n]}(\rho)).$$

\end{proof}

\end{theorem}

Theorems \ref{th:prand} and  \ref{th:and} (which have been proved
for the canonical holistic conjunctions ${\mathtt{AND}}^{(m,n)}$) can
be easily generalized to any ${\mathtt{AND}}_\mathfrak T^{(m,n)}$
(where $\mathfrak T$ is any truth-perspective).

\section{A Holistic quantum computational semantics}
 Let us first present the syntactical basis for our  semantics. The linguistic framework
   is a
    {\em quantum computational
  language\/} $\mathcal L$, whose alphabet
   contains atomic formulas (say, ``the spin-value
in the $x$-direction is up''), including two privileged formulas
$\bf t$ and $\bf f$ that represent the truth-values {\em Truth\/}
and {\em Falsity\/}, respectively. The connectives of $\mathcal L$
correspond to some gates
  that have a special logical and computational interest:
 the negation $\lnot$ (corresponding to the gate {\em negation\/}),
  a ternary connective
$\intercal$ (corresponding to the {\em Toffoli-gate\/}), the
exclusive disjunction $\uplus$ (corresponding to ${\mathtt{XOR}}$), the
square root of the identity $\sqrt{id}$ (corresponding to the
 {\em Hadamard-gate\/}), the square root of negation
$\sqrt{\lnot}$ (corresponding to the gate {\em square root of
${\mathtt{NOT}}$\/}). The notion of {\em formula\/} (or {\em
sentence\/}) of $\mathcal L$ is inductively defined (in the
expected way).
 Accordingly, if
$\alpha$, $\beta$, $\gamma$ are formulas, then  the expressions
$\lnot \alpha$, $\sqrt{id}\, \alpha$, $\sqrt{\lnot}\, \alpha$,
$\intercal(\alpha, \beta,\, \gamma)$, $\alpha\uplus \beta$ are
formulas.

Recalling the definition of the holistic conjunction ${\mathtt{AND}}^{(m,n)}$, it is useful  to introduce  a binary logical
conjunction $\land$ by means of the following metalinguistic
definition:
$$ \alpha \land \beta := \intercal(\alpha, \beta, \mathbf f)$$
(where the  false formula $\mathbf f$ plays  the role of a {\em
syntactical ancilla\/}).

On this basis, a (binary) inclusive disjunction is
(metalinguistically) defined via de Morgan-law:
$$ \alpha \lor \beta := \lnot(\lnot \alpha \land \lnot \beta).$$

The connectives $\lnot$, $\land$, $\lor$ and $\uplus$  will be
also termed {\em quantum computational Boolean connectives}; while
$\sqrt{id}$  and $\sqrt{\lnot}$  represent {\em genuine quantum
computational connectives}. A formula that contains at most
Boolean connectives  is  called a {\em Boolean formula\/} of
$\mathcal L$.

In the following we will use $\mathbf q,\mathbf q_1, \ldots$ as
metavariables for atomic formulas, while $\alpha,\beta,
\gamma,\ldots $ will represent generic formulas.

 \begin{definition} {\em (The atomic complexity of a formula)}
  \label{de:atcomp}\nl
 The atomic complexity $At(\alpha)$ of a formula $\alpha$ is
  the number of occurrences of atomic formulas in
 $\alpha$.
 \end{definition}

 For instance, $At(\intercal(\mathbf q, \mathbf q,\mathbf f)) = 3$.
 The notion of atomic complexity plays an
 important semantic  role. As we will see, the meaning of any formula whose
 atomic complexity is $n$ shall live in the domain
 $\mathfrak D(\mathcal H^{(n)})$. For this reason,
 $\mathcal H^{(At(\alpha))}$ (briefly indicated by $\mathcal
 H^\alpha$) will be also called the {\em semantic space\/} of $\alpha$.

 Any formula
$\alpha$ can be naturally decomposed into its parts, giving rise
to a special configuration called the {\em syntactical tree\/} of
$\alpha$ (indicated by {$STree^{\alpha}$}). Roughly,
$STree^{\alpha}$ can be represented as a finite sequence of {\em
levels\/}:
\begin{align*}
&Level_h^\alpha\\
&\ldots \ldots \\
&Level_1^\alpha
\end{align*} where:

\begin{itemize}
\item each $Level_i^\alpha$ (with $1 \le i \le h$) is a
    sequence   $(\beta_1, \ldots,  \beta_m)$  of subformulas
    of $\alpha$;
    \item the {\em bottom level\/} $Level_1^\alpha$ is
        $(\alpha)$; \item the {\em top level\/ }
        $Level_h^\alpha$ is the sequence $(\mathbf q_1,
        \ldots, \mathbf q_r)$, where $\mathbf q_1, \ldots,
        \mathbf q_r$ are the atomic occurrences in $\alpha$;
        \item for any $i$ (with $1 \le i < h$),
            $Level_{i+1}^\alpha$ is the  sequence obtained by
            dropping the {\em principal connective\/} in all
            molecular formulas  occurring at $Level_i^\alpha$,
            and by repeating all the atomic sentences that
            occur at $Level_i^\alpha$.
\end{itemize}

By {\em Height\/} of $\alpha$ (indicated by ${Height(\alpha)}$) we
mean the number $h$ of levels of the syntactical tree of $\alpha$.

As an example, consider the following formula: $$\alpha=
\lnot\intercal(\mathbf q,\lnot \mathbf q,\mathbf f) =\lnot(\mathbf
q \land \lnot \mathbf q),$$ which represents an instance of the non-contradiction principle.

The syntactical tree of $\alpha$ is the following sequence of
levels:
\begin{align*}
Level_4^\alpha &= (\mathbf q, \mathbf q, \mathbf f)\\
Level_3^\alpha &= (\mathbf q,\lnot \mathbf q;\mathbf f)\\
Level_2^\alpha &=
(\intercal(\mathbf q,\lnot \mathbf q,\mathbf f))\\
Level_1^\alpha &=
(\lnot\intercal(\mathbf q,\lnot \mathbf q,\mathbf f))
\end{align*}
Clearly, $Height(\alpha)= 4.$

For any choice of a truth-perspective $\mathfrak T$, the
syntactical tree of any formula $\alpha$   uniquely determines a
sequence of gates, all defined on the semantic space of $\alpha$.
As an example, consider again the formula $\alpha= \lnot
\intercal(\mathbf q,\lnot \mathbf q,\mathbf f)$. In the
syntactical tree of $\alpha$ the third level has been obtained
from the fourth level by repeating the first occurrence of
$\mathbf q$, by negating the second occurrence of $\mathbf q$ and
by repeating $\mathbf f$, while the second and the first level
have been obtained  by applying, respectively,
 the connectives $\intercal$ and $\lnot$ to formulas occurring at
 the levels
immediately above.

Accordingly, one can say that, for any choice of a
truth-perspective $\mathfrak T$, the syntactical tree of $\alpha$
uniquely determines the following sequence consisting of three
gates, all defined on the semantic space of $\alpha$:
$$\left(^\mathfrak D{\mathtt{I}}_\mathfrak T^{(1)}
\otimes\, ^\mathfrak D{\mathtt{NOT}}_\mathfrak T^{(1)}
\otimes \,^\mathfrak D{\mathtt{I}}_\mathfrak T^{(1)},\quad
^\mathfrak D{\mathtt{T}}^{(1,1,1)}_\mathfrak T, \quad ^\mathfrak D{\mathtt{NOT}}^{(3)}_\mathfrak T\right).$$ Such a sequence is called the
$\mathfrak T$-{\em gate tree\/} of $\alpha$. This procedure can be
naturally generalized to any formula $\alpha$. The general form of
the $\mathfrak T$- gate tree of  $\alpha$ will be:
  $$(^{\mathfrak D}G^{\alpha}_{\mathfrak T_{(h-1)}},
\ldots, ^{\mathfrak D}G^{\alpha}_{\mathfrak T_{(1)}}),$$
 where $h$ is the Height of
$\alpha$.

From an intuitive point of view, any formula $\alpha$ of $\mathcal
L$ can be regarded as a synthetic logical description of a quantum
circuit that may assume as inputs qumixes living in the semantic
space of $\alpha$. For instance, the circuit described by $\alpha
= \lnot \intercal(\mathbf q, \lnot \mathbf q, \, \mathbf f)$ can
be represented as follows:

\begin{figure}[h]
\centering
\includegraphics[scale=0.4]{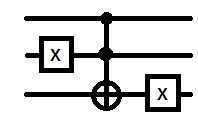}
\end{figure}\nl
Thus, $\mathcal L$-formulas turn out to have a characteristic {\em
dynamic\/} character, representing  systems of {\em computational
actions\/}.

Now the holistic semantics comes into play.\footnote{See
\cite{DFGLLS} and \cite{DBGS}.} The intuitive idea can be sketched
as follows. For any choice of a truth-perspective, a {\em holistic
model\/}  of the language $\mathcal L$ assigns to any formula
$\alpha$ a {\em global informational meaning\/} that lives in
$\mathcal H^\alpha$ (the semantic space of $\alpha$). This meaning
determines the {\em contextual meanings\/} of all subformulas of
$\alpha$ (from the whole to the parts!). It may happen that one
and the same model assigns to a given formula $\alpha$ different
contextual meanings in different contexts.

Before defining the concept of model, it is expedient to introduce
the weaker notion of {\em holistic map\/}  for the language
$\mathcal L$.

\begin{definition} {\em (Holistic map)} \label{de:holmap}\nl
 A {\em holistic map} for $\mathcal L$ (associated to a
 truth-perspective $\mathfrak T$) is a
map ${\mathtt{Hol}}_\mathfrak T$ that assigns a
 meaning ${\mathtt{Hol}}_\mathfrak T(Level_i^\alpha)$ to each level of the
syntactical tree of $\alpha$, for any formula $\alpha$. This
meaning is a qumix living in the semantic space of $\alpha$.

\end{definition}

Given a formula $\gamma$, any holistic map ${\mathtt{Hol}}_\mathfrak T$
determines the {\em contextual
 meaning\/},
 with respect to the context ${\mathtt{Hol}}_\mathfrak T(\gamma)$,
 of any occurrence of a subformula $\beta$
in $\gamma$. This contextual meaning can be defined, in a natural
way, by using the notion of {\em reduced  state\/}.

\begin{definition}{\em (Contextual meaning)} \label{de:cont} \nl
Consider a formula $\gamma$ such that $Level_i^\gamma
=(\beta_{i_1},\ldots, \beta_{i_r})$. We have: $\mathcal H^\gamma=
\mathcal H^{\beta_{i_1}}\otimes \ldots \otimes \mathcal
H^{\beta_{i_r}} $. Let ${\mathtt{Hol}}_\mathfrak T$ be a holistic map.
The {\em contextual meaning\/} of the occurrence $\beta_{i_j}$
with respect to the context ${\mathtt{Hol}}_\mathfrak T(\gamma)$ is
defined as follows:
$${\mathtt{Hol}}_\mathfrak T^{\gamma}(\beta_{i_j}):=
 Red^j_{[At(\beta_{i_1}),\ldots, At(\beta_{i_r})]}
 ({\mathtt{Hol}}_\mathfrak T(Level_i(\gamma))).$$ \end{definition}
Of course, we obtain:
$${\mathtt{Hol}}_\mathfrak T^{\gamma}(\gamma)= {\mathtt{Hol}}_\mathfrak T(\gamma).$$

A holistic map ${\mathtt{Hol}}_\mathfrak T$ is called {\em normal
for a formula $\gamma$\/} iff for any subformula $\beta$ of
$\gamma$, ${\mathtt{Hol}}_\mathfrak T$ assigns the same contextual
meaning to all occurrences of $\beta$ in the syntactical tree of
$\gamma$. In other words:
$${\mathtt{Hol}}_\mathfrak T^\gamma(\beta_{i_j})=
{\mathtt{Hol}}_\mathfrak T^\gamma(\beta_{u_v}),$$ where
$\beta_{i_j}$ and $\beta_{u_v}$  are two occurrences of $\beta$ in
$STree^\gamma$.\nl A {\em normal holistic map\/} is a holistic map
${\mathtt{Hol}}_\mathfrak T$ that is normal for all formulas
$\gamma$.

\begin{definition}{\em (Compositional holistic map)}\label{de:compos} \nl\nl
Consider  a formula $\alpha$ such that $Level_h^\alpha = (\mathbf
q_1, \ldots, \mathbf q_r)$, while the $\mathfrak T$-gate tree of
$\alpha$ is $(^\mathfrak D G_{\mathfrak T_{(h-1)}},
\ldots,\,^\mathfrak D G_{\mathfrak T_{(1)}})$. A holistic map
${\mathtt{Hol}}_\mathfrak T$ is called {\em compositional with respect
to $\alpha$} iff  the following conditions are satisfied:
\begin{enumerate}
\item[(1)] ${\mathtt{Hol}}_\mathfrak T(Level_h^\alpha) = {\mathtt{Hol}}_\mathfrak T^\alpha(\mathbf q_1) \otimes \ldots
    \otimes  {\mathtt{Hol}}_\mathfrak T^\alpha(\mathbf q_r)$.
\item[(2)] ${\mathtt{Hol}}_\mathfrak T(Level_i^\alpha) =\,
    ^\mathfrak D G_{\mathfrak T_{(i)}}({\mathtt{Hol}}_\mathfrak
    T(Level_{i+1}^\alpha)) $, for any $i$ \nl (with $1 \le i <
    h$).

\end{enumerate}
\end{definition}

\begin{lemma}\label{le:compos} \nl \nl
Any holistic map $\mathtt{Hol}_\mathfrak T$ that is compositional
with respect to the formula $\alpha$ satisfies the following
conditions:
\begin{enumerate}
\item[(1)] If $Level_i^\alpha = (\beta_{i_1}, \ldots,
    \beta_{i_r})$, then $$\mathtt{Hol}_\mathfrak
    T(Level_i^\alpha)= \mathtt{Hol}^\alpha_\mathfrak
    T(\beta_{i_1}) \otimes \ldots \otimes \mathtt{    Hol}^\alpha_\mathfrak T(\beta_{i_r}),
    $$ for any $i$ such that $1 \le i \le Height(\alpha)$.
\item[(2)] $\mathtt{Hol}_\mathfrak T$ is a normal holistic map for $\alpha$.
\begin{proof}\nl \nl
\item[(1)] By definition of compositional holistic map and
by induction on $i$.
\item[(2)] By definition of compositional holistic map and
by (1).

\end{proof}

\end{enumerate}

\end{lemma}

We  can now define the concept of {\em holistic
 model\/} of the language $\mathcal L$.

\begin{definition} {\em (Holistic model)} \label{de:model}\nl
 A {\em holistic model} of  $\mathcal L$ is a normal
 holistic map ${\mathtt{Hol}}_\mathfrak T$  that satisfies the following
conditions for any formula $\alpha$.
\begin{enumerate}
  \item[(1)] Let $(^\mathfrak D G^\alpha_{\mathfrak
      T_{(h-1)}}, \ldots,\, ^\mathfrak D G^\alpha_{\mathfrak
      T_{(1)}}) $ be the $\mathfrak T$-gate tree of $\alpha$
      and let $1 \le i < h$. Then,
$$ {\mathtt{Hol}}_\mathfrak T(Level_{i}^\alpha)= \,
^\mathfrak D G^\alpha_{\mathfrak T_{(i)}}
 ({\mathtt{Hol}}_\mathfrak T(Level_{i+1}^\alpha)).$$
 In other words the  meaning of each level (different from the
 top level) is obtained by applying the corresponding gate to
 the meaning of the level that occurs immediately above.

 \item [(2)] Let $Level_i^\alpha = (\beta_{i_1}, \ldots,
     \beta_{i_r}).$ Then,\nl $\beta_{i_j} = \mathbf f
     \Rightarrow {\mathtt{Hol}}_\mathfrak T^\alpha(\mathbf f)=
     Red^j_{[At(\beta_{i_1}),\ldots, At(\beta_{i_r})]}({\mathtt{Hol}}_\mathfrak T(Level_{i}^\alpha))=\, ^\mathfrak T
     P_0^{(1)};$ \nl $\beta_{i_j} = \mathbf t \Rightarrow
{\mathtt{Hol}}_\mathfrak T^\alpha(\mathbf t)=
Red^j_{[At(\beta_{i_1}),\ldots,
At(\beta_{i_r})]}({\mathtt{Hol}}_\mathfrak
T(Level_{i}^\alpha))= \, ^\mathfrak T P_1^{(1)}$, for any $j$
$(1 \le j \le r)$.\nl In other words, the contextual meanings
of $\mathbf f$ and of $\mathbf t$ are always the $\mathfrak
T$-{\em falsity\/} and the $\mathfrak T$-{\em truth\/},
respectively.

\end{enumerate}

 \end{definition}

 On this basis, we put:
$${\mathtt{Hol}}_\mathfrak T(\alpha):=
 {\mathtt{Hol}}_\mathfrak T(Level_1^\alpha), $$
 for any formula $\alpha$.

 Since all gates are reversible, assigning a value ${\mathtt{Hol}}_\mathfrak T(Level_i^\alpha)$ to a particular
 $Level_i^\alpha$ of $STree^\alpha$ determines  the value
 ${\mathtt{Hol}}_\mathfrak T(Level_j^\alpha)$ for any other level $Level_j^\alpha$.
 Consequently, ${\mathtt{Hol}}_\mathfrak T(Level_i^\alpha)$ determines the contextual
 meaning ${\mathtt{Hol}}_\mathfrak T^\alpha(\beta)$ for any subformula $\beta$ of
 $\alpha$.

Notice that any ${\mathtt{Hol}}_\mathfrak T(\alpha)$ represents a kind
of autonomous semantic context that is not necessarily correlated
with the meanings of other formulas. Generally we have:
 $$
 {\mathtt{Hol}}_\mathfrak
T^{\gamma}(\beta) \neq
  {\mathtt{Hol}}_\mathfrak
T^{\delta}(\beta).
 $$
 Thus, one and the same formula may receive different contextual
 meanings in different contexts (as, in fact, happens in the case of
 our normal use of natural languages).

 \begin{definition}{\em (Compositional holistic model)} \label{compmod}\nl
 \nl
 A holistic model ${\mathtt{Hol}}_\mathfrak T$ is called
 \begin{itemize}
 \item {\em compositional\/} iff ${\mathtt{Hol}}_\mathfrak T$ is a
     holistic map that is compositional with respect to all
     formulas $\alpha$;
\item {\em perfectly compositional\/} iff
    ${\mathtt{Hol}}_\mathfrak T$ is a compositional model that
    satisfies the following condition for any formulas
    $\alpha$, $\beta$ and for any atomic formula $\mathbf q$
    (occurring in $\alpha$ and in $\beta$):
    $${\mathtt{Hol}}^\alpha_\mathfrak T(\mathbf q)=
    {\mathtt{Hol}}^\beta_\mathfrak T(\mathbf q).   $$

 \end{itemize}

 \end{definition}
 Accordingly, models that are  perfectly compositional  are
 context-independent; while compositional models may be
 context-dependent. As expected, the {\em compositional quantum
 computational semantics\/}, that  only refers to compositional models
 (or to perfectly compositional models), represents a special case
 of the holistic quantum computational semantics.

Consider now a formula $\alpha$  whose atomic complexity is $n$.
By definition of model we have: ${\mathtt{Hol}}_\mathfrak T(\alpha) \in
\mathfrak D(\mathcal H^{(n)})$. From an intuitive point of view,
the qumix $Red^n_{[1,\ldots,n]}({\mathtt{Hol}}_\mathfrak T(\alpha))$
(which lives the space $\C^2$) can be regarded as a {\em
generalized truth-value\/} of $\alpha$ (determined by the model
${\mathtt{Hol}}_\mathfrak T$). At the same time, the number ${\mathtt{p}}_\mathfrak T({\mathtt{Hol}}(\alpha))$ represents the
probability-value of $\alpha$ with respect to the
truth-perspective $\mathfrak T$ (determined by the model ${\mathtt{Hol}}_\mathfrak T$). Accordingly, our semantics can be described as
a {\em two-level many valued semantics\/}, where for any choice of
a model ${\mathtt{Hol}}_\mathfrak T$, any formula receives two
correlated {\em semantic values\/}: a generalized truth-value
(represented by a density operator of $\C^2$) and a
probability-value (a real number in the interval $[0,1]$).

To what extent do contextual meanings and gates (associated to the
logical connectives) commute? In this respect the 1-ary
connectives ($\lnot$, $\sqrt{id}$ and $\sqrt{\lnot}$) behave
differently from the binary and the ternary connectives ($\uplus$
and $\intercal$).
\begin{theorem} \label{th:commut} \nl \nl
Consider a model $\mathtt{Hol}_\mathfrak T$.
\begin{enumerate}
\item[(1)] Let $\lnot \beta$ be a subformula of $\gamma$.
    Then, $$\mathtt{Hol}^\gamma_\mathfrak T(\lnot \beta) = \,
    ^\mathfrak D\mathtt{NOT}_\mathfrak T^{(At(\beta))}(\mathtt{    Hol}^\gamma_\mathfrak T(\beta) ).$$
\item[(2)] Let $\sqrt{id} \beta$ be a subformula of $\gamma$.
    Then,
    $$\mathtt{Hol}^\gamma_\mathfrak T(\sqrt{id} \beta) = \,
    ^\mathfrak D\sqrt{\mathtt{I}}_\mathfrak T^{(At(\beta))}
    (\mathtt{Hol}^\gamma_\mathfrak T(\beta) ).$$
\item[(3)] Let $\sqrt{\lnot} \beta$ be a subformula of
    $\gamma$. Then,
    $$\mathtt{Hol}^\gamma_\mathfrak T(\sqrt{\lnot} \beta) = \,
    ^\mathfrak D\sqrt{\mathtt{NOT}}_\mathfrak T^{(At(\beta))}
    (\mathtt{Hol}^\gamma_\mathfrak T(\beta) ).$$
\end{enumerate}
In other words, the contextual meaning of the negation of a
formula $\beta$ can be obtained by applying the appropriate
negation-gate to the contextual meaning of $\beta$. In a  similar
way for the connectives $\sqrt{id}$ and  $\sqrt{\lnot}$.
\end{theorem}
\begin{proof} By definition of syntactical tree, of $\mathfrak T$-gate
tree, of holistic model and of contextual meaning.

\end{proof}

Such a commutativity-situation breaks down in the case of the
binary and ternary connectives ($\uplus$, $\intercal$). As we have
seen, the conjunction ${\mathtt{AND}}_\mathfrak T^{(m,n)}$ has a
characteristic holistic behavior. Generally, we have:\nl
${\mathtt{AND}}_\mathfrak T^{(m,n)}(\rho) = \,^\mathfrak
D{\mathtt{T}}_\mathfrak T^{(m,n,1)} (\rho \otimes\, ^\mathfrak T
P^{(1)}_0)\,\neq$ \nl $ ^\mathfrak D{\mathtt{T}}_\mathfrak
T^{(m,n,1)} (Red^{(1)}_{[m,n,1]}(\rho) \otimes
Red^{(2)}_{[m,n,1]}(\rho)\otimes
 \, ^\mathfrak TP^{(1)}_0).  $

 Consequently, from a semantic point of view, we will generally
 obtain:
$${\mathtt{Hol}}_\mathfrak T^\gamma
(\intercal (\alpha, \beta, \mathbf f)) \neq
\, ^\mathfrak D{\mathtt{T}}_\mathfrak T^{(At(\alpha),
 At(\beta), At(\mathbf f))}
({\mathtt{Hol}}_\mathfrak T^\gamma(\alpha)
\otimes {\mathtt{Hol}}_\mathfrak T^\gamma(\beta)\otimes
{\mathtt{Hol}}_\mathfrak T^\gamma(\mathbf f)).$$

A similar situation holds for the binary connective $\uplus$.

The connectives $\intercal$ and $\uplus$ satisfy a weaker
relation, described by the following theorem.

\begin{theorem} \label{th:weakcomm}  \nl \nl
Consider a  model ${\mathtt{Hol}}_\mathfrak T$.
\begin{enumerate}

\item[(1)] Let $\intercal (\alpha_1,\alpha_2,\alpha_3) $ be a
    subformula of $\gamma$. Thus, the syntactical tree of
    $\gamma$  contains two levels having the following
    form:\nl
    \begin{itemize}
     \item $Level^\gamma_{(i+1)}=\,
    (\beta_{(i+1)_1},\ldots, \beta_{(i+1)_{k_1}},
    \beta_{(i+1)_{k_2}}, \beta_{(i+1)_{k_3}},\ldots,
    \beta_{(i+1)_{r}}),$ where
 $\alpha_1= \beta_{(i+1)_{k_1}}$, $\alpha_2=
    \beta_{(i+1)_{k_2}}$, $\alpha_3= \beta_{(i+1)_{k_3}}$.
    \item $Level^\gamma_{i}=\, (\beta_{i_1},\ldots,
        \beta_{i_{j}}, \ldots, \beta_{i_{s}}),$ where
 $\intercal (\alpha_1,\alpha_2,\alpha_3) = \beta_{i_{j}}$.

    \end{itemize}

     We have:\nl $\mathtt{Hol}_\mathfrak T^\gamma(\intercal
    (\alpha_1,\alpha_2,\alpha_3)) = $ \nl $^\mathfrak D \mathtt{    T}_\mathfrak T^{(At(\alpha_1),At(\alpha_2), At(\alpha_3))}
    (Red^{(k_1, k_2,k_3)}_{[At(\beta_{(i+1)_{1}}),\ldots,
    At(\beta_{(i+1)_{r}})]}(\mathtt{Hol}_\mathfrak
    T(Level_{(i+1)}(\gamma))))$.
\item[(2)] A similar relation holds  when $\alpha_1 \uplus
    \alpha_2 $ is a subformula of $\gamma$.

\end{enumerate}

\end{theorem}
\begin{proof} By definition of syntactical tree, of $\mathfrak T$-gate
tree, of holistic model and of contextual meaning.

\end{proof}

The holistic behavior of the connectives $\intercal$ and $\uplus$
seem to reflect pretty well (at a semantic level) the holistic
behavior of quantum circuits. As is well known, trying to separate
the different branches ``inside the box'' of a given quantum
computation generally has the effect of destroying the
characteristic parallelism (and hence the efficiency) of the
computation in question.

The following Lemma will play an important role in the development
of the holistic semantics.

\begin{lemma}\label{le:noia}\nl \nl
Consider a formula $\gamma$ and  let $\eta$ be a subformula of
$\gamma$. For any model $\mathtt{Hol}_\mathfrak T$  and for any
formula $\beta$ there exists a model $^*\mathtt{Hol}_\mathfrak T$
such that,
$$^*\mathtt{Hol}_\mathfrak T^{\gamma \land \beta}(\eta) =
\mathtt{Hol}_\mathfrak T^\gamma(\eta).$$
\end{lemma}

\begin{proof}({\em Sketch})
 Consider two formulas $\gamma$ and $\beta$
 and let ${\mathtt{Hol}}_\mathfrak T$ be a model. If $\beta$ is a
 subformula of $\gamma$ the proof is trivial (since it is
 sufficient to take
 $^*\mathtt{Hol}_\mathfrak T$ equal to
$\mathtt{Hol}_\mathfrak T$). Suppose that $\beta$ is not a
subformula of $\gamma$
 (while $\gamma$ and $\beta$  may have some
common subformulas).  Consider the syntactical tree of $\gamma
\land \beta$, which includes (in its left part) the syntactical
tree of $\gamma$ (where $Level_1^\gamma$ appears at
$Level_2^{\gamma \land \beta}$, while the top level of
$STree^{\gamma}$ is supposed to be repeated until the Height $h$
of $STree^{\gamma \land \beta}$ is reached). The model
${\mathtt{Hol}}_\mathfrak T$ assigns a qumix
${\mathtt{Hol}}_\mathfrak T(Level_i^\gamma)$ to each level of
$STree^\gamma$ (represented as a part of $STree^{\gamma \land
\beta}$). Let us briefly write: $^\gamma \rho_{i+1} = \,
{\mathtt{Hol}}_\mathfrak T(Level_i^\gamma)$. We transform $STree
^{\gamma \land \beta}$ into a ``hybrid'' object $Hybr$ that is a
sequence of sequences $Hybr_i$. Each $Hybr_i$ corresponds to
$Level_i^{\gamma \land \beta}$ and is a sequence of objects that
are either formulas or qumixes. Taking into account the fact the
$\intercal(\gamma, \beta, \mathbf f)$ and $\beta$ are not
subformulas of $\gamma$, we define the first two elements of
$Hybr$ as follows:\nl $Hybr_1= (\intercal(\gamma, \beta, \mathbf
f));\quad Hybr_2= (^\gamma \rho_2,\beta, \, ^\mathfrak T
P_0^{(1)}).$ Then, we proceed (step by step) by replacing the
first occurrence in $STree^{\gamma \land \beta}$ of each formula
$\theta$ that is also a subformula of $\gamma$ with the qumix
${\mathtt{Hol}}^\gamma_\mathfrak T(\theta)$. Suppose, for
instance, that $\theta$ occurs for the first time at
$Level_i^{\gamma \land \beta}$, and suppose that $\theta=
\intercal(\xi_1, \xi_2, \xi_3)$. Then (by definition of
syntactical tree), $\xi_1$, $\xi_2$ and $\xi_3$ shall occur at
$Level_{i+1}^{\gamma \land \beta}$. We define $Hybr_i$ and
$Hybr_{i+1}$ in such a way that the following conditions are
satisfied: a) in $Hybr_i$  the qumix ${\mathtt{Hol}}_\mathfrak
T^\gamma(\theta)$ occurs in place of the formula $\theta$
(occurring at $Level_{i}^{\gamma \land \beta}$); b) in
$Hybr_{i+1}$ the qumix $[^\mathfrak D{\mathtt{T}}_\mathfrak
T^{(At(\xi_1), At(\xi_2), At(\xi_3))}]^{-1}
({\mathtt{Hol}}^{\gamma}_\mathfrak T(\theta)) $ occurs in place of
the subsequence $(\xi_1, \xi_2, \xi_3)$ (occurring at
$Level_{i+1}^{\gamma \land \beta}$). We proceed in a similar way
for all possible linguistic forms of $\theta$. When we finally
reach the top level $Level_h^{\gamma \land \beta}$, the
corresponding $Hybr_h$ will have the following form:
 $$Hybr_h = (^\gamma \rho_h, Ob_1, \ldots, Ob_t, \,^\mathfrak T P_0^{(1)}),
    $$
 where each $Ob_j$ is either a qumix or an atomic formula
 $\mathbf q$ that does not occur in $\gamma$. Now, we replace in
 $Hybr_h$
 each ``surviving'' formula $\mathbf q$ with the qumix
 ${\mathtt{Hol}}_\mathfrak T(\mathbf q)$
 (which lives in $\C^2$).
This operation destroys   the  ``hybrid'' form of $Hybr_h$, which
is now transformed into a  homogeneous sequence
 of qumixes:
 $$
 ^\mathfrak D
 Hybr_h = (^\gamma \rho_h, \, ^\mathfrak D Ob_1, \ldots, \,
 ^\mathfrak D Ob_t, \,^\mathfrak T P_0^{(1)}), \quad\text{where}:$$
$$
^\mathfrak D Ob_j = \begin{cases}
 Ob_j, \,\text{if}\, Ob_j \, \text{is a qumix}; \\
 {\mathtt{Hol}}_\mathfrak T(\mathbf q),\, \text{if}\,
  Ob_j = \mathbf q.
\end{cases}
$$
On this basis, we transform the whole  $Hybr$ into a sequence of
qumix-sequences $ ^\mathfrak D Hybr_i$. Let us first refer to
$Hybr_{h-1}$, which may contain
 formulas that are not subformulas of $\gamma$. Suppose,
for instance, that the first formula occurring in $ Hybr_{h-1}$ is
$$
\beta_{{(h-1)}_j}=\intercal(\mathbf q_1, \mathbf q_2, \mathbf q_3).
$$
Since $\beta_{{(h-1)}_j}$ is not a subformula of $\gamma$, $
^\mathfrak D Hybr_{h}$ shall contain three separate qumixes
$^{\mathbf q_1}\!\rho,\,^{\mathbf q_2}\!\rho,\,^{\mathbf
q_3}\!\rho$ (corresponding to the atom-sequence $(\mathbf q_1,
\mathbf q_2, \mathbf q_3)$ occurring in the right part of
$STree^{\gamma \land \beta}$). On this basis, we replace the
formula $\intercal(\mathbf q_1, \mathbf q_2, \mathbf q_3)$ with
the qumix $^\mathfrak D{\mathtt{T}}_\mathfrak T^{(1,1,1)}
(^{\mathbf q_1}\!\rho\,\otimes\, ^{\mathbf
q_2}\!\rho\otimes\,^{\mathbf q_3}\!\rho) $ in $ Hybr_{h-1}$ and in
all other $ Hybr_i$ where $\intercal(\mathbf q_1, \mathbf q_2,
\mathbf q_3)$ possibly appears.

Then, we proceed step by step by applying the same procedure to
all formulas $\beta_{i_j}$ occurring in $ Hybr_i$, for  any $i$
($1\le i < h$). At the end of the procedure, each $Hybr_i$ ($1 <
i\le h$) has been transformed into a sequence of qumixes
$$
^\mathfrak D Hybr_i=
(^\gamma\!\rho_i,\rho_{i_1},\ldots,\rho_{i_r},^\mathfrak T\!\!P_0^{(1)}),
$$
where any qumix is naturally associated to a segment of
$Level_i^{\gamma\land\beta}$.

We define now  the  map    $^*{\mathtt{Hol}}_\mathfrak T$ in the following way:
\begin{itemize}
\item $^*{\mathtt{Hol}}_\mathfrak T(Level_i^{\gamma\land\beta})=
 \,^\gamma\!\rho_i\otimes\rho_{i_1}\otimes\ldots\otimes\rho_{i_r}\otimes\,^\mathfrak T\!P_0^{(1)},
 \, \text{if}\,\,\, 1< i\le h$;
\item $^*{\mathtt{Hol}}_\mathfrak
    T(Level_i^{\gamma\land\beta}) =\,\, ^\mathfrak
    D{\mathtt{T}}_\mathfrak T^{(At(\gamma), At(\beta), 1)}
(^*{\mathtt{Hol}}_\mathfrak T(Level_2^{\gamma\land\beta} ))$,
    if $i=1$.
\end{itemize}
We have:
 \begin{enumerate}
 \item[(I)] by construction, $^*{\mathtt{Hol}}_\mathfrak T(Level_i^{\gamma\land\beta})$
is a qumix of $\mathcal H^{\gamma\land\beta}$. Hence, $^*{\mathtt{Hol}}_\mathfrak T$ is
a holistic map for $\gamma\land\beta$;
\item[(II)] $^*{\mathtt{Hol}}_\mathfrak T$ is normal for
    $\gamma\land\beta$, by the normality of
    ${\mathtt{Hol}}_\mathfrak T$ and because different
    occurrences in $Hybr$ of a formula that is not a
    subformula of $\gamma$ have been replaced by the same
    qumix;
\item[(III)] by construction, $^*{\mathtt{Hol}}_\mathfrak T$ preserves the logical form of all
subformulas of $\gamma\land\beta$. Accordingly,
$^*{\mathtt{Hol}}_\mathfrak T(Level_i^{\gamma \land \beta}) =
 \,^\mathfrak D G_{\mathfrak T_{(i)}}
 (^*{\mathtt{Hol}}_\mathfrak T
 (Level_{i+1}^{\gamma \land \beta}))$, for any $i$ such that
 $1 \le i < h$, where  $(^\mathfrak D G_{\mathfrak T_{(h-1)}},
\ldots,\,  ^\mathfrak D G_{\mathfrak T_{(1)}})$ is the
 $\mathfrak T$-gate tree of $\gamma \land \beta$. Furthermore,
 the sentences $\mathbf f$ and $\mathbf t$ have (trivially)
 the ``right'' contextual meanings. Hence,
 $^*{\mathtt{Hol}}_\mathfrak T$ is a model for
 $\gamma\land\beta$;
\item[(IV)] by construction, for any $\eta$ that is a subformula of $\gamma$:
$$^*\mathtt{Hol}_\mathfrak T^{\gamma \land \beta}(\eta) =
\mathtt{Hol}_\mathfrak T^\gamma(\eta).$$
\end{enumerate}
\end{proof}

Now the concepts of {\em truth\/}, {\em validity\/}, {\em logical
consequence\/} and {\em logical equivalence\/} can be  defined in
terms of the probability-function ${\mathtt{p}}_\mathfrak T$ and of the
preorder $\preceq_\mathfrak T$.

\begin{definition} {\em (Truth)} \label{de:truth}\nl
\nl A formula $\alpha$ is called {\em true}  with respect to a
model ${\mathtt{Hol}}_\mathfrak T$  (abbreviated as $\models_{{\mathtt{Hol}}_\mathfrak T} \alpha$)  iff  ${\texttt p}_\mathfrak T({\mathtt{Hol}}_\mathfrak T(\alpha)) = 1.$ \end{definition}

\begin{definition} {\em (Validity)} \label{de:val}\nl
\begin{enumerate}
\item[1)] $\alpha$ is called $\mathfrak T$-{\em valid}
    ($\models_\mathfrak T \alpha$) iff for any model ${\mathtt{Hol}}_\mathfrak T$, $\models_{{\mathtt{Hol}}_\mathfrak T}
    \alpha$.
\item[2)]$\alpha$ is called {\em valid} ($\models \alpha$) iff
    for any truth-perspective $\mathfrak T$,
    $\models_\mathfrak T \alpha$.
    \end{enumerate}

\end{definition}

\begin{definition} {\em (Logical consequence)} \label{de:cons}\nl
\nl \begin{enumerate} \item[1)] $\beta$ is called a $\mathfrak
T$-{\em logical consequence\/}
      of $\alpha $ ($\alpha \vDash_\mathfrak T \beta$)
       iff for any
      formula $\gamma$ such that $\alpha$ and $\beta$ are
      subformulas of $\gamma$ and
      for any
     model ${\mathtt{Hol}}_\mathfrak T$,
  $${\mathtt{Hol}}_\mathfrak T^\gamma(\alpha) \preceq_\mathfrak T
   {\mathtt{Hol}}_\mathfrak T^\gamma (\beta).$$
  \item[2)] $\beta$ is called a {\em logical consequence\/}
      of $\alpha $ ($\alpha \vDash\beta$)
       iff for  any truth-perspective $\mathfrak T$,
       $\alpha \vDash_\mathfrak T \beta$.

   \end{enumerate}\end{definition}

   When
       $\alpha \vDash_{\mathtt{I}} \beta$, we  say that $\beta$ is a
       {\em canonical logical consequence\/} of $\alpha$.

 \begin{definition} {\em (Logical equivalence)} \nl
\nl $\alpha$ and $\beta$ are logically equivalent ($\alpha \equiv
\beta$) iff $\alpha \vDash \beta$ and $\beta \vDash \alpha$.
  \end{definition}

  The concept of logical consequence turns out to be  invariant
  with respect  to truth-perspective changes.

  \begin{lemma}\cite{DBGS}\label{le:invariant}\nl\nl
  $\alpha \vDash \beta$ iff  $\alpha \vDash_{\mathtt{I}} \beta$ iff
  there is a truth-perspective $\mathfrak T$ such that
  $\alpha \vDash_{\frak T} \beta$.

  \end{lemma}

 Although the holistic semantics is strongly context-dependent,
 one can prove that  the logical consequence-relation is reflexive
  and transitive.

\begin{theorem}\label{th:trans}\nl \nl
\begin{enumerate}
\item[(1)] $\alpha \vDash \alpha$;
\item[(2)] $\alpha \vDash \beta \quad \text{and}\quad \beta
    \vDash \delta \quad \Rightarrow \quad \alpha \vDash
    \delta $.\end{enumerate}

\end{theorem}
\begin{proof}\nl

\begin{enumerate}
\item[(1)] Straightforward.
\item[(2)]  Assume the hypothesis and suppose, by
    contradiction, that there exist a model ${\mathtt{Hol}}_\mathfrak T$ and a formula $\gamma$, where $\alpha$
    and $\delta$ occur as subformulas, such that: ${\mathtt{Hol}}^\gamma_\mathfrak T(\alpha)\, \not\preceq_\mathfrak
    T\, {\mathtt{Hol}}^\gamma_\mathfrak T(\delta)$. Consider the
    formula $\gamma \land \beta$. By Lemma \ref{le:noia} there
    exists a model $^*{\mathtt{Hol}}_\mathfrak T$ such that for any $\eta$ that
    is a subformula of $\gamma$: $^*{\mathtt{Hol}}_\mathfrak T^{\gamma \land \beta}(\eta) = {\mathtt{Hol}}_\mathfrak T^{\gamma}(\eta)$. Thus, we have: \nl $^*{\mathtt{Hol}}_\mathfrak T^{\gamma \land \beta}(\alpha) = {\mathtt{Hol}}_\mathfrak T^{\gamma }(\alpha)$ and $^*{\mathtt{Hol}}_\mathfrak T^{\gamma \land \beta}(\delta) = {\mathtt{Hol}}_\mathfrak T^{\gamma }(\delta)$ .\nl Since we have
    assumed (by contradiction) that ${\mathtt{Hol}}_\mathfrak
    T^{\gamma }(\alpha) \not\preceq_\mathfrak T {\mathtt{Hol}}_\mathfrak T^{\gamma }(\delta)$, we obtain: $^*{\mathtt{Hol}}_\mathfrak T^{\gamma \land \beta }(\alpha)
    \not\preceq_\mathfrak T \,^*{\mathtt{Hol}}_\mathfrak
    T^{\gamma \land \beta}(\delta)$, against the hypothesis
    and the transitivity of $\preceq_\mathfrak T$, which
    imply:\nl $^*{\mathtt{Hol}}_\mathfrak T^{\gamma \land
    \beta}(\alpha)
     \preceq_\mathfrak T \, ^*{\mathtt{Hol}}_\mathfrak T^{\gamma
    \land \beta }(\beta)$; $^*{\mathtt{Hol}}_\mathfrak T^{\gamma
    \land \beta}(\beta) \preceq_\mathfrak T
     \,^*{\mathtt{Hol}}_\mathfrak T^{\gamma \land \beta
     }(\delta)$; \nl $^*{\mathtt{Hol}}_\mathfrak T^{\gamma
     \land \beta}(\alpha)
    \preceq_\mathfrak T \,^*{\mathtt{Hol}}_\mathfrak T^{\gamma
    \land \beta }(\delta)$.

\end{enumerate}

  \end{proof}

The concept of logical consequence, defined in this semantics,
characterizes a special form of quantum computational logic
(formalized in the language $\mathcal L$) that is termed {\em
holistic quantum computational logic\/} ({\bf HQCL}). One can
easily show that {\bf HQCL} includes, as a particular fragment,
classical sentential logic (representing also an adequate
description of classical circuits).

 Consider the
sublanguage $\mathcal L^C$ of $\mathcal L$, whose formulas are the
Boolean formulas of $\mathcal L$.
\begin{definition}{\em (Classical quantum computational model)}
 \label{de:classmod} \nl\nl
A {\em classical quantum computational model\/}
is a model ${\mathtt{Hol}}_\mathfrak T$ that satisfies the following conditions:
\begin{enumerate}
\item[(1)] $\mathfrak T$ is the canonical truth-perspective
    ${\mathtt{I}}$.
\item [(2)] ${\mathtt{Hol}}_\mathfrak T$ is only defined for
    $\mathcal L^C$-formulas.
\item[(3)] For any formula $\alpha$ of $\mathcal L^C$,
    ${\mathtt{Hol}}_\mathfrak T$  assigns to the top level of
    the syntactical tree of $\alpha$ a (canonical) register
    (living in the semantic space of $\alpha$).

\end{enumerate}

\end{definition}

One immediately obtains that any classical quantum computational
model assigns to any Boolean formula of $\mathcal L$ a (canonical)
register (living in its semantic space).

 We can now define a
  consequence-relation that concerns  the Boolean language
$\mathcal L^C$.

\begin{definition} {\em (Classical quantum computational consequence)}
 \label{de:classcons}\nl
\nl A formula $\beta$ of $\mathcal L^C$ is called a {\em classical
quantum computational consequence\/}
      of a formula $\alpha $ of $\mathcal L^C$ ($\alpha \vDash_{CQC} \beta$)
       iff for any
      formula $\gamma$ of $\mathcal L^C$ such that $\alpha$ and $\beta$ are
      subformulas of $\gamma$ and
      for any classical quantum computational
    model ${\mathtt{Hol}}_{\mathtt{I}}$,
  $${\mathtt{Hol}}_{\mathtt{I}}^\gamma(\alpha) \preceq_{\mathtt{I}}
   {\mathtt{Hol}}_{\mathtt{I}}^\gamma (\beta).$$

   \end{definition}

   \begin{lemma}\label{le:coincid} \nl\nl
   For any  formulas $\alpha$ and $\beta$ of $\mathcal L^C$,
    $\alpha \vDash_{CQC} \beta$  iff $\beta$ is a logical
   consequence of $\alpha$ according to classical sentential logic.

   \end{lemma}
   \begin{proof} Straightforward. \end{proof}

\section{Logical arguments}
Which logical arguments are valid or are possibly violated in the
logic {\bf HQCL}? The following theorems give some answers to this
question. By Lemma \ref{le:invariant} it will be sufficient to
refer to the canonical logical consequence relation and to
canonical models. Accordingly, we will write: ${\mathtt{p}}$,
$\preceq$, ${\mathtt{Hol}}$ and $\vDash$ (instead of
${\mathtt{p}}_{\mathtt{I}}$, $\preceq_{\mathtt{I}}$,
${\mathtt{Hol}}_{\mathtt{I}}$ and $\vDash_{\mathtt{I}}$).

  Theorem \ref{th:valbool} sums up some basic arguments that
  hold for the quantum computational Boolean connectives.

  \begin{theorem}\label{th:valbool}\nl
  \begin{enumerate}
  \item [(1)] $\alpha \land \beta \vDash \alpha$; $\alpha
      \land \beta \vDash \beta$
  \item [(2)] $\alpha \vDash \beta \, \Rightarrow \,
      \alpha \land \delta \vDash \beta$
  \item [(3)] $\lnot \lnot \alpha \equiv \alpha$
  \item [(4)] $ \alpha \vDash \beta \, \Rightarrow \,\lnot
      \beta \vDash \lnot \alpha $
  \item [(5)] $\mathbf f \vDash \beta$;  $\beta \vDash \mathbf
      t$

  \end{enumerate}
  \end{theorem}
\begin{proof}\nl \nl
\begin{enumerate}
\item [(1)] $\alpha\land\beta\vDash\alpha$;\,
    $\alpha\land\beta\vDash\beta$.\nl Let $\alpha$ and $\alpha
    \land \beta$ (= $\intercal (\alpha,\beta, \mathbf f)$)  be
    subformulas of $\gamma$. Suppose that $\alpha$, $\beta$,
    $\mathbf f$ occur respectively at the positions $k_1$,
    $k_2$, $k_3$ of $Level_{i+1}^\gamma$ (in the syntactical
    tree of $\gamma$), while  $\intercal (\alpha,\beta,
    \mathbf f)$ occurs at $Level_{i}^\gamma$. By Theorem
    \ref{th:weakcomm}(1),  for any ${\mathtt{Hol}}$ we have:
    $${\mathtt{Hol}}^\gamma(\intercal (\alpha,\beta,\mathbf f))= \,
    ^\mathfrak D{\mathtt{T}}^{(At(\alpha),At(\beta),
    At(\mathbf f))}(Red_{[1,\ldots, r]}^{(k_1,k_2,k_3)}
    ({\mathtt{Hol}}(Level^\gamma_{i+1})))   $$
    (where $r$ is the number of formulas occurring at
    $Level_{i+1}^\gamma$). Hence, by definition of contextual
    meaning and by Theorem \ref{th:and}:
    $${\mathtt{Hol}}^\gamma(\intercal (\alpha,\beta,\mathbf f)) \preceq
    {\mathtt{Hol}}^\gamma (\alpha).   $$
    In a similar way one can prove that $\alpha \land \beta
      \vDash \beta$.

\item [(2)] $\alpha\vDash\beta\;\Rightarrow \;\alpha\land
    \delta \vDash \beta$.\nl Assume the hypothesis and let
    $\alpha\land \delta$, $\beta$ be subformulas of $\gamma$.
Then $\alpha$ and $\delta$  also are subformulas of $\gamma$.
By hypothesis, for any ${\mathtt{Hol}}$: ${\mathtt{Hol}}^\gamma(\alpha)
\preceq  {\mathtt{Hol}}^\gamma(\beta)$. By (1): ${\mathtt{Hol}}^\gamma(\alpha \land \delta) \preceq {\mathtt{Hol}}^\gamma(\alpha)$. Hence, by transitivity of $\preceq$:
${\mathtt{Hol}}^\gamma(\alpha \land \delta) \preceq  {\mathtt{Hol}}^\gamma(\beta)$.

\item [(3)] $\lnot\lnot\alpha\equiv\alpha$.\nl Let $\lnot
    \lnot \alpha$ and $\alpha$ be subformulas of $\gamma$. By
    Theorem \ref{th:commut} (1) and by the double-negation
    principle for the gate $^\mathfrak D {\mathtt{NOT}}^{(n)}$, we
    obtain for any ${\mathtt{Hol}}$:
    $${\mathtt{Hol}}^\gamma(\lnot\lnot\alpha) =
    \,^\mathfrak D {\mathtt{NOT}}^{(At(\alpha))} \,
    ^\mathfrak D {\mathtt{NOT}}^{(At(\alpha))}{\mathtt{Hol}}^\gamma(\alpha) =
    {\mathtt{Hol}}^\gamma(\alpha).$$

\item [(4)]
    $\alpha\vDash\beta\;\Rightarrow\;\lnot\beta\vDash\lnot\alpha$.\nl
Assume the hypothesis and let $\lnot \beta$, $\lnot \alpha$ be
subformulas of $\gamma$. Then $\alpha$ and $\beta$ also are
subformulas of $\gamma$. By hypothesis, for any ${\mathtt{Hol}}$,
${\mathtt{p}}({\mathtt{Hol}}^\gamma(\alpha))\leq {\mathtt{p}}({\mathtt{Hol}}^\gamma (\beta))$. Hence, $1-{\mathtt{p}}({\mathtt{Hol}}^\gamma(\beta))\leq 1- {\mathtt{p}}({\mathtt{Hol}}^\gamma(\alpha))$.
Since for any $\rho \in \mathfrak D(\mathcal H^{(n)}),\, {\mathtt{p}}(^\mathfrak D{\mathtt{NOT}}^{(n)}\rho) = 1 - {\mathtt{p}}(\rho)$, we
obtain:\nl $^\mathfrak D{\mathtt{NOT}}^{(At(\beta))}{\mathtt{Hol}}^\gamma(\beta)\, \preceq \,^\mathfrak D{\mathtt{NOT}}^{(At(\alpha))}{\mathtt{Hol}}^\gamma(\alpha)$. Whence, by
Theorem \ref{th:commut} (1), ${\mathtt{Hol}}^{\gamma}(\lnot \beta)
\preceq {\mathtt{Hol}}^{\gamma}(\lnot \alpha)$.\nl
\item [(5)]$\mathbf f \vDash\beta$; $\beta\vDash\mathbf t$.\nl
    Let $\beta$ and $\mathbf f$  be subformulas of $\gamma$.
    By definition of holistic model we have: $ {\mathtt{p}}({\mathtt{Hol}}^\gamma(\mathbf f)) = {\mathtt{p}}(P^{(1)}_0) = 0$, for any
    ${\mathtt{Hol}}$. Hence, ${\mathtt{Hol}}^\gamma(\mathbf f) \preceq
    {\mathtt{Hol}}^\gamma(\beta)$. In a similar way one proves that
    $\beta\vDash\mathbf t$.
    \end{enumerate}

  \end{proof}

The dual forms of 5.1(1) and of  5.1(2) hold for the connective
$\lor$.

The following theorem sums up some significant classical arguments
that are not valid for the quantum computational Boolean
connectives.

\begin{theorem}\label{th:nval}\nl \nl
\begin{enumerate}
\item [(1)]$\alpha \nvDash \alpha \land \alpha$
\item  [(2)]$\alpha \land \beta \nvDash \beta \land \alpha$
\item [(3)] $\alpha \land (\beta \land \delta) \nvDash (\alpha
    \land \beta) \land \delta$
\item [(4)] $(\alpha \land \beta) \land \delta \nvDash \alpha
    \land (\beta \land \delta)$
\item [(5)] $\alpha \land (\beta \lor \delta) \nvDash (\alpha
    \land \beta) \lor (\alpha \land \delta)$
\item [(6)] $(\alpha \land \beta) \lor (\alpha \land \delta)
    \nvDash \alpha \land (\beta \lor \delta)$
\item [(7)] $\delta \vDash \alpha  \,\text{and}\, \delta
    \vDash \beta \,\nRightarrow \, \delta \vDash \alpha
    \land \beta $
\item [(8)] $\alpha \land \lnot \alpha \nvDash \beta$
\item [(9)] $\alpha \uplus \beta \nvDash \beta \uplus \alpha$
\item [(10)] $\alpha \uplus \beta \nvDash \alpha \lor \beta;
    \quad\alpha \uplus \beta \nvDash \lnot \alpha \lor \lnot
    \beta$

\end{enumerate}

\end{theorem}

\begin{proof} In the following counterexamples
 $\alpha$, $\beta$ and $\delta$ will always represent atomic formulas.\nl \nl

\begin{enumerate}
\item [(1)] $\alpha \nvDash \alpha \land \alpha$ \nl Take
    $\gamma=\alpha\land\alpha$ and consider a model ${\mathtt{Hol}}$ such that  ${\mathtt{Hol}}(\gamma)=\, ^\mathfrak D {\textrm{
 T}}^{(1,1,1)} (\frac{1}{2} {\mathtt{I}}\otimes\frac{1}{2} {\mathtt{I}}\otimes P^{(1)}_0)$.
We have: $\Prob(\Hol{\alpha})=\frac{1}{2} >
\Prob(\Hol{\alpha\land\alpha})=\frac{1}{4}$.

\item [(2)] $\alpha \land \beta \nvDash \beta \land
    \alpha$.\nl Take $\gamma=
    (\alpha\land\beta)\land(\beta\land\alpha)$. Consider a
    model ${\mathtt{Hol}}$ such that  ${\mathtt{Hol}}(\gamma)=
    \,^\mathfrak D {\textrm{T}}^{(3,3,1)}[^\mathfrak D {\textrm{
    T}}^{(1,1,1)} (\frac{1}{2} {\mathtt{I}}\otimes\frac{1}{2} {\mathtt{I}}\otimes P^{(1)}_0) \otimes \, ^\mathfrak D {\textrm{
 T}}^{(1,1,1)} (P_{\frac{1}{\sqrt{2}}(\ket{01}-\ket{10})}
 \otimes P^{(1)}_0)\otimes P^{(1)}_0].$ We have:
 $\Prob(\Hol{\alpha\land\beta})= \frac{1}{4} >
\Prob(\Hol{\beta\land\alpha})=0.$

\item [(3)] $\alpha \land (\beta \land \delta) \nvDash (\alpha
    \land \beta) \land \delta$. \nl Take
    $\gamma=(\alpha\land(\beta\land\delta))\land
((\alpha\land\beta)\land\delta)$ and consider a model ${\mathtt{Hol}}$ such that ${\mathtt{Hol}}(\gamma)= \, ^\mathfrak D {\textrm{
T}}^{(5,5,1)} [\,^\mathfrak D {\textrm{T}}^{(1,3,1)} (\frac{1}{2}
 {\mathtt{I}}\otimes\,^\mathfrak D {\textrm{T}}^{(1,1,1)} (\frac{1}{2}
 {\mathtt{I}}\otimes\frac{1}{2} {\mathtt{I}}\otimes P^{(1)}_0) \otimes
 P_0^{(1)}) \otimes \, ^\mathfrak D {\textrm{T}}^{(3,1,1)}(\,
  ^\mathfrak D {\textrm{T}}^{(1,1,1)} \otimes {\mathtt{I}}^{(2)}
(P_{\frac{1}{\sqrt{2}}(\ket{01010}-\ket{10000})})) \otimes
P_0^{(1)}]$. \\
We have:
$\Prob(\Hol{\alpha\land(\beta\land\delta)})=\frac{1}{8}
> \Prob(\Hol{(\alpha\land\beta)\land\delta})=0$.

\item [(4)] $ (\alpha \land \beta) \land \delta \nvDash \alpha
    \land (\beta \land \delta)$. \nl Similar to (3).

\item [(5)] $\alpha \land (\beta \lor \delta) \nvDash (\alpha
    \land \beta) \lor (\alpha \land \delta)$.\nl Take
    $\gamma=(\alpha \land (\beta \lor \delta)) \land ((\alpha
\land \beta) \lor (\alpha \land \delta))$.\nl Consider a model
${\mathtt{Hol}}$ such that ${\mathtt{Hol}}(\gamma)=$ \nl $ ^\mathfrak D
{\textrm{
 T}}^{(5,7,1)} [^\mathfrak D {\textrm{T}}^{(1,3,1)} (\frac{1}{2}
 {\mathtt{I}} \otimes \, ^\mathfrak D \Not^{(3)}\, ^\mathfrak D{\textrm{
 T}}^{(1,1,1)} (P_{\frac{1}{\sqrt{2}}(\ket{01}-\ket{10})}
\otimes P^{(1)}_0) \otimes P^{(1)}_0)\otimes \, ^\mathfrak D
\Not^{(7)}\, ^\mathfrak D{\textrm{T}}^{(3,3,1)} ( \, ^\mathfrak D
\Not^{(3)}\, ^\mathfrak D{\textrm{T}}^{(1,1,1)} (\frac{1}{2} {\mathtt{I}}\otimes\frac{1}{2} {\mathtt{I}}\otimes P^{(1)}_0)\otimes \,
 ^\mathfrak D \Not^{(3)}\, ^\mathfrak D{\textrm{T}}^{(1,1,1)}
 (\frac{1}{2} {\mathtt{I}}\otimes\frac{1}{2} {\mathtt{I}}\otimes
P^{(1)}_0)\otimes P^{(1)}_0) \otimes P^{(1)}_0] $. \nl We
have: $\Prob(\Hol{\alpha \land (\beta \lor \delta)})=
\frac{1}{2} > \Prob(\Hol{(\alpha \land \beta) \lor (\alpha
\land \delta)}) =\frac{7}{16}$.

\item [(6)] $(\alpha \land \beta) \lor (\alpha \land \delta)
    \nvDash \alpha \land (\beta \lor \delta)$.\nl Take
    $\gamma=(\alpha \land (\beta \lor \delta)) \land ((\alpha
    \land \beta) \lor (\alpha \land \delta))$.\nl Consider a
model ${\mathtt{Hol}}$ such that  ${\mathtt{Hol}}(\gamma)=$ \nl
$^\mathfrak D {\textrm{T}}^{(5,7,1)} [ \, ^\mathfrak D {\textrm{T}}^{(1,3,1)} (\frac{1}{2} {\mathtt{I}}\otimes \, ^\mathfrak  D
\Not^{(3)} \,^\mathfrak D{\textrm{T}}^{(1,1,1)} (\frac{1}{2} {\mathtt{I}}\otimes\frac{1}{2} {\mathtt{I}}\otimes P^{(1)}_0) \otimes
P^{(1)}_0)\otimes \, ^\mathfrak D \Not^{(7)}\,^\mathfrak D
{\textrm{T}}^{(3,3,1)} ( \, ^\mathfrak D \Not^{(3)} \,^\mathfrak
D{\textrm{T}}^{(1,1,1)} (\frac{1}{2} {\mathtt{I}}\otimes\frac{1}{2} {\mathtt{I}}\otimes P^{(1)}_0)\otimes \, ^\mathfrak D \Not^{(3)}
\,^\mathfrak D{\textrm{T}}^{(1,1,1)} (\frac{1}{2} {\mathtt{I}}\otimes\frac{1}{2} {\mathtt{I}}\otimes P^{(1)}_0)\otimes
P^{(1)}_0) \otimes P^{(1)}_0] $. \nl We have:
$\Prob(\Hol{(\alpha \land \beta) \lor (\alpha \land
\delta)})=\frac{7}{16} > \Prob(\Hol{\alpha \land (\beta \lor
\delta)})=\frac{3}{8}$.

\item [(7)]
    $\delta\vDash\alpha\;\text{and}\;\delta\vDash\beta\;
    \nRightarrow\;\delta\vDash\alpha\land\beta $.\nl Take
$\gamma=(\alpha\land\beta)\land\delta$. Consider a model ${\mathtt{Hol}}$ such that ${\mathtt{Hol}}(\gamma)= \, ^\mathfrak D {\mathtt{T}}^{(3,1,1)} (\, ^\mathfrak D {\mathtt{T}}^{(1,1,1)} (\frac{1}{2}
 {\mathtt{I}}\otimes\frac{1}{2} {\mathtt{I}}\otimes
 P^{(1)}_0)\otimes\frac{1}{2} {\mathtt{I}}\otimes P^{(1)}_0)$. We
have: $\Prob(\Hol{\alpha})= \Prob(\Hol{\beta})=
\Prob(\Hol{\delta})=\frac{1}{2} >
\Prob(\Hol{\alpha\land\beta})=\frac{1}{4}$.

\item [(8)] $\alpha \land \lnot \alpha \nvDash \beta$.\nl Take
    $\gamma=(\alpha\land\lnot\alpha)\land\beta$.\nl Consider a
    model ${\mathtt{Hol}}$ such that ${\mathtt{Hol}}(\gamma)=\,^\mathfrak D {\textrm{T}}^{(3,1,1)}
    (\,^\mathfrak D {\textrm{T}}^{(1,1,1)} (\frac{1}{2} {\mathtt{I}}\otimes\frac{1}{2} {\mathtt{I}}\otimes P^{(1)}_0)\otimes
P^{(1)}_0\otimes P^{(1)}_0)$. \nl We have:
$\Prob(\Hol{\alpha\land\lnot\alpha})=\frac{1}{4}$ and
$\Prob(\Hol{\beta})=0$.

\item [(9)] $\alpha \uplus \beta \nvDash \beta \uplus
    \alpha$.\nl Take
    $\gamma=(\alpha\uplus\beta)\land(\beta\uplus\alpha)$.
    Consider a model ${\mathtt{Hol}}$ such that ${\mathtt{Hol}}(\gamma)=\, ^\mathfrak D {\mathtt{T}}^{(2,2,1)} [\,
    ^\mathfrak D {\mathtt{XOR}}^{(1,1)} P_{\frac{1}
 {\sqrt{2}}(\ket{01}-\ket{10})} \otimes \, ^\mathfrak D {\mathtt{XOR}}^{(1,1)} (\frac{1}{2} {\mathtt{I}}\otimes\frac{1}{2} {\mathtt{I}})\otimes P^{(1)}_0]$.  We have:
$\Prob(\Hol{\alpha\uplus\beta})=1 >
\Prob(\Hol{\beta\uplus\alpha})=\frac{1}{2}$.

\item [(10)] $\alpha \uplus \beta \nvDash \alpha \lor \beta$;
    $\alpha \uplus \beta \nvDash \lnot \alpha \lor \lnot
    \beta$.\nl Take
    $\gamma=(\alpha\uplus\beta)\land(\alpha\lor\beta)$.
    Consider a model $\mathtt{Hol}$ such that ${\mathtt{Hol}}(\gamma)=\,
^\mathfrak D {\mathtt{T}}^{(2,3,1)} [\, ^\mathfrak D {\mathtt{XOR}}^{(1,1)} P_{\frac{1}{\sqrt{2}}(\ket{01}-\ket{10})}
    \otimes \, ^\mathfrak D \Not^{(3)} {\mathtt{T}}^{(1,1,1)}
    (\Not^{(1)}\otimes\Not^{(1)}\otimes {\mathtt{I}}^{(1)})
    (\frac{1}{2} {\mathtt{I}}\otimes\frac{1}{2} {\mathtt{I}}\otimes
    P^{(1)}_0)\otimes P^{(1)}_0]$. \nl We have:
    $\Prob(\Hol{\alpha\uplus\beta})=1 >
    \Prob(\Hol{\alpha\lor\beta})=\frac{3}{4}$.\nl In a similar
    way one proves that $\alpha \uplus \beta \nvDash \lnot
    \alpha \lor \lnot \beta$.

\end{enumerate}\end{proof}

The dual forms of 5.2(1)-5.1(8) hold for the connective $\lor$.

  The following theorem sums up some basic arguments that
  hold for the genuine quantum computational connectives.

\begin{theorem}\label{th:genui}\nl \nl
\begin{enumerate}
\item [(1)]$\sqrt{id}\sqrt{id} \alpha \equiv \alpha$
  \item [(2)]$\sqrt{id}\mathbf f \equiv \sqrt{id}\mathbf t$
  \item [(3)]$\lnot \sqrt{id}\mathbf f \equiv \sqrt{id}\mathbf
      f$; $\lnot \sqrt{id}\mathbf t \equiv \sqrt{id}\mathbf t$
  \item [(4)]$\sqrt{id}(\alpha \land \beta) \equiv \sqrt{id}
      \mathbf f$
\item [(5)]$\sqrt{\lnot}\sqrt{\lnot} \alpha \equiv
    \lnot\alpha$
  \item [(6)]$\sqrt{\lnot}\mathbf f \equiv \sqrt{\lnot}\mathbf
      t$
  \item [(7)]$\lnot \sqrt{\lnot}\mathbf f \equiv
      \sqrt{\lnot}\mathbf f$; $\lnot \sqrt{\lnot}\mathbf t
      \equiv \sqrt{\lnot}\mathbf t$
  \item [(8)] $\lnot \sqrt{\lnot} \alpha \equiv \sqrt{\lnot}\,
      \lnot \alpha$
  \item [(9)] $\sqrt{\lnot}(\alpha \land \beta) \equiv
      \sqrt{\lnot} \mathbf f$
\item [(10)]$\sqrt{id}\sqrt{\lnot}\alpha \equiv
    \sqrt{id}\alpha$
\item [(11)]$\sqrt{\lnot}\sqrt{id}\alpha \equiv
    \lnot\sqrt{\lnot}\alpha$
\item [(12)]$\sqrt{id}\sqrt{\lnot}(\alpha \land \beta) \equiv
    \sqrt{\lnot} \mathbf f$
\item [(13)] $\sqrt{\lnot}\sqrt{id}(\alpha \land \beta) \equiv
    \sqrt{\lnot} \mathbf f$
\end{enumerate}

  \end{theorem}
\begin{proof}\nl \nl
\begin{enumerate}
\item [(1)] $\sqrt{id}\sqrt{id} \alpha \equiv \alpha$.\\
Let $\alpha$ and $\sqrt{id}\sqrt{id} \alpha$ be subformulas of
    $\gamma$. Since for any $\rho \in \mathfrak D(\mathcal
    H^{(n)})$, $^\mathfrak D\sqrt{{\mathtt{I}}}^{(n)}\, ^\mathfrak
    D\sqrt{{\mathtt{I}}}^{(n)} \rho =\rho$, by Theorem
    \ref{th:commut}(2) we obtain: ${\mathtt{Hol}}^\gamma(\sqrt{id}\sqrt{id} \alpha)= {\mathtt{Hol}}^\gamma(\alpha)$.

\item [(2)] $\sqrt{id}\mathbf f \equiv \sqrt{id}\mathbf t$\\
By definition of model, by Theorem \ref{th:commut}(2) and
because $\Prob(\, ^\mathfrak D \SId^{(1)}
P_0^{(1)})=\frac{1}{2}= \Prob(\, ^\mathfrak D \SId^{(1)}
P_1^{(1)})$.

\item [(3)] $\lnot \sqrt{id}\mathbf f \equiv \sqrt{id}\mathbf
    f$;
      $\lnot \sqrt{id}\mathbf t \equiv \sqrt{id}\mathbf t$\\
Let $\lnot \sqrt{id}\mathbf f$ and $\sqrt{id}\mathbf f$ be
subformulas of $\gamma$. By definition of model and by Theorem
\ref{th:commut}(1,2) we have: \nl $\Prob(\Hol{\lnot
\sqrt{id}\mathbf f})= \Prob(\, ^\mathfrak D \Not^{(1)}
\,^\mathfrak D\SId^{(1)} P_0^{(1)})=\frac{1}{2} = \Prob(\,
^\mathfrak D \SId^{(1)} P_0^{(1)}) =
{\mathtt{Hol}}^\gamma(\sqrt{id}\mathbf f)$. In a similar way
one proves that $\lnot \sqrt{id}\mathbf t \equiv
\sqrt{id}\mathbf t$.
\item [(4)] $\sqrt{id}(\alpha \land \beta) \equiv \sqrt{id} \mathbf f$.\\
Let $\sqrt{id}(\alpha \land \beta)$ and $\sqrt{id} \mathbf f$
be subformulas of $\gamma$. Let $At(\alpha)= m$ and
$At(\beta)= n$. Suppose that in the syntactical tree of
$\gamma$ the subformulas $\alpha$ and $\beta$ occur,
respectively, at the positions $k_1$ and $k_2$ of
$Level_i^\gamma$ (consisting of $r$ formulas). Consider a
model ${\mathtt{Hol}}$ and let $\rho = \,
Red^{(k_1,k_2)}_{[1,\ldots,r]}({\mathtt{Hol}}(Level_i^\gamma))$.
By Theorem \ref{th:weakcomm}(1), we have:
${\mathtt{Hol}}^\gamma (\alpha \land \beta) = \, ^\mathfrak
D{\mathtt{T}}^{(m,n,1)}(\rho \otimes P_0^{(1)})$. Then, by
definition of ${\mathtt{p}}$ and by Lemma \ref{le:tof} we
obtain:\nl
 $\Prob(\Hol{\sqrt{id}(\alpha \land \beta)})=
\Tr[P_1^{(m+n+1)}\, ^\mathfrak D \SId^{(m+n+1)}\,^\mathfrak D
{\textrm{T}}^{(m,n,1)} (\rho\otimes P_0^{(1)})]= $ \\
$\Tr[({\mathtt{I}}^{(m+n)}-P_1^{(m)}\otimes
P_1^{(n)})\rho({\mathtt{I}}^{(m+n)}-P_1^{(m)}\otimes
P_1^{(n)})\otimes P_1^{(1)}\SId^{(1)}P_0^{(1)}\SId^{(1)}+
(P_1^{(m)}\otimes P_1^{(n)})\rho (P_1^{(m)}\otimes
P_1^{(n)})\otimes
P_1^{(1)}\SId^{(1)}\Not^{(1)}P_0^{(1)}\Not^{(1)}\SId^{(1)}]=$\\
$\Tr[({\mathtt{I}}^{(m+n)}-P_1^{(m)}\otimes
P_1^{(n)})\rho({\mathtt{I}}^{(m+n)}-P_1^{(m)}\otimes
P_1^{(n)})] \Tr(P_1^{(1)}\SId^{(1)}P_0^{(1)}\SId^{(1)})+$\\
$\Tr[(P_1^{(m)}\otimes P_1^{(n)})\rho (P_1^{(m)}\otimes
P_1^{(n)})] \Tr(P_1^{(1)}\SId^{(1)}P_1^{(1)}\SId^{(1)})=$\\
$\Tr[({\mathtt{I}}^{(m+n)}-P_1^{(m)}\otimes
P_1^{(n)})\rho({\mathtt{I}}^{(m+n)}-P_1^{(m)}\otimes
P_1^{(n)})] \frac{1}{2}+$ \\ $ \Tr[(P_1^{(m)}\otimes
P_1^{(n)})\rho (P_1^{(m)}\otimes P_1^{(n)})] \frac{1}{2}=
\Tr(\rho)\frac{1}{2}=\frac{1}{2} =
{\mathtt{p}}(\sqrt{id}\mathbf f)$.

\item [(5)]$\sqrt{\lnot}\sqrt{\lnot} \alpha \equiv \lnot\alpha$.\\
By Theorem  \ref{th:commut}(1,3) and because $\sqrt{{\mathtt{NOT}}}^{(n)} \sqrt{{\mathtt{NOT}}}^{(n)}\rho = {\mathtt{NOT}}^{(n)}\rho$,
for any $\rho \in \mathfrak D(\mathcal H^{(n)})$.

\item [(6)] $\sqrt{\lnot}\mathbf f \equiv \sqrt{\lnot}\mathbf t$.\\
Similar to (2).

\item [(7)] $\lnot \sqrt{\lnot}\mathbf f \equiv
    \sqrt{\lnot}\mathbf f$; $\lnot \sqrt{\lnot}\mathbf t
      \equiv \sqrt{\lnot}\mathbf t$.\\
By definition of model, by Theorem \ref{th:commut}(1,3) and
because\nl $\Prob(\, ^\mathfrak D \Not^{(1)}\,^\mathfrak D
\SNot^{(1)} P_0^{(1)})=\frac{1}{2} = \Prob(\,^\mathfrak
D\SNot^{(1)} P_0^{(1)})= \Prob(\, ^\mathfrak D \Not^{(1)}
\,^\mathfrak D\SNot^{(1)} P_1^{(1)}) = \Prob(\,^\mathfrak
D\SNot^{(1)} P_1^{(1)})$.

\item [(8)] $\lnot \sqrt{\lnot} \alpha \equiv \sqrt{\lnot}\,
    \lnot \alpha$.\\
    By Theorem \ref{th:commut}(1,3) and because $^\mathfrak
    D\sqrt{{\mathtt{NOT}}}^{(n)} \,^\mathfrak D\sqrt{{\mathtt{NOT}}}^{(n)} \rho = \,^\mathfrak D{\mathtt{NOT}}^{(n)}\rho$, for
    any $\rho \in \mathfrak D(\mathcal H^{(n)})$.

\item [(9)] $\sqrt{\lnot}(\alpha \land \beta) \equiv \sqrt{id}
    \mathbf f$.
\\
Similar to (4).

\item [(10)]$\sqrt{id}\sqrt{\lnot}\alpha \equiv \sqrt{id}\alpha$.\\
By Theorem \ref{th:commut} (2,3) and because $\Prob(^\mathfrak
D \SId^{(n)}\,^\mathfrak D \SNot^{(n)} \rho)= \Prob(
^\mathfrak D \SId^{(n)} \rho)$ for any $\rho \in \mathfrak
D(\mathcal H^{(n)})$.

\item [(11)] $\sqrt{\lnot}\sqrt{id}\alpha \equiv \lnot\sqrt{\lnot}\,\alpha$.\\
Similar to (10).

\item [(12)] $\sqrt{id}\sqrt{\lnot}(\alpha \land \beta) \equiv \sqrt{\lnot}
\, \mathbf f$\\
Similar to (4).

\item [(13)] $\sqrt{\lnot}\sqrt{id}(\alpha \land \beta) \equiv \sqrt{\lnot}
\, \mathbf f$\\
Similar to (4).

\end{enumerate}
\end{proof}

Theorem \ref{th:nval} shows how the ``Boolean'' fragment of  {\bf
HQCL} (formalized in the language $\mathcal L^C$) is a quite weak
logic with strongly non-classical features. As happens in the case
of most fuzzy logics, conjunctions and disjunctions of {\bf HQCL}
are generally non-idempotent. This is, of course, expected in all
situations that involve information-transmission, where ``repetita
iuvant''. The failure of commutativity, associativity and
distributivity (for $\land$ and $\lor$) seems to be confirmed and
justified by a number of examples that concern informal arguments
(expressed in  natural languages) or semantic situations arising
in the languages of art (for instance, in literature or in music).
At the same time, it is not easy  to find appropriate linguistic
models for the genuine quantum computational connectives
($\sqrt{id}$ and $\sqrt{\lnot}$) outside the domain of
quantum-information phenomena. Interestingly enough, a suggestion,
in this direction, comes from a formal semantics of
music.\footnote{See \cite{DGLN}. } Consider the case of {\em
musical modulations\/}, whose characteristic role is creating
tonality-changes in a given composition. The starting point may be
a {\em precise\/} tonality (say, $C$ major), which is followed by
a situation of {\em tonal ambiguity\/}, where different tonalities
co-exist in a form that seems to behave like a quantum
superposition. Finally one arrives at another tonality, which may
be either closely related or distant with respect to the original
tonality. From an abstract point of view, such a transformation
seems to be similar to what happens when the gate square root of
negation is applied twice. A first application of $\sqrt{{\mathtt{NOT}}}^{(1)}$ to a classical certainty, represented for instance by
the bit $\ket{1}$,  gives rise to a maximally uncertain {\em
quantum perhaps\/}: the superposition $\frac{1}{\sqrt{2}}(\ket{0}
+ \ket{1})$, for which the two canonical truth-values are equally
probable. But, then,  a second application of the same gate
transforms this maximal uncertainty into a different classical
certainty, represented by the bit $\ket{0}$.

\end{document}